\documentclass[a4paper,10pt]{article}

\usepackage{amsmath,amsfonts,amssymb,graphicx,enumerate,verbatim,color,wasysym}
\usepackage{algorithm}
\usepackage{algpseudocode}
\usepackage{natbib}
\usepackage{amsthm}
\usepackage{hyperref}

\newcommand\bigfrac[2]{\frac{\displaystyle{#1}}{\displaystyle{#2}}}
\renewcommand{\vec}[1]{\boldsymbol{#1}}

\newtheorem{thm}{Theorem}

\title{\LARGE Modelling and analysis of rank ordered data with ties via a generalized
  Plackett-Luce model}
\author{Daniel A.\ Henderson} 
\date{School of Mathematics, Statistics \& Physics, Newcastle
  University, NE1 7RU, U.K.\\ \href{mailto:daniel.henderson@newcastle.ac.uk}{daniel.henderson@newcastle.ac.uk}}

\begin{document}

\maketitle

\begin{abstract}

  A simple generative model for rank ordered data with ties is
  presented. The model is based on ordering geometric latent variables
  and can be seen as the discrete counterpart of the Plackett-Luce
  (PL) model, a popular, relatively tractable model for permutations.
  The model, which will be referred to as the GPL model, for
  generalized (or geometric) Plackett-Luce model, contains the PL
  model as a limiting special case. A closed form expression for the
  likelihood is derived.  With a focus on Bayesian inference via data
  augmentation, simple Gibbs sampling and EM algorithms are derived
  for both the general case of multiple comparisons and the special
  case of paired comparisons. The methodology is applied to several
  real data examples. The examples highlight the flexibility of the
  GPL model to cope with a range of data types, the simplicity and
  efficiency of the inferential algorithms, and the ability of the GPL
  model to naturally facilitate predictive inference due to its simple
  generative construction. 

\smallskip

 \noindent \textit{Key words: Bayesian inference; bucket order; EM algorithm; Gibbs
   sampler; latent variables; MCMC algorithms; ordered partitions; prediction}

\end{abstract}

\section{Introduction}
\label{sec:intro}

Data in the form of rankings and orderings are ubiquitous and there
are many models for such data \citep{Marden95,AlvoY14}. The
Plackett-Luce (PL) model \citep{Luce59,Plackett75} is one such model
for rank ordered data that has proved popular due to its relative
tractability, intuitive interpretation, and its flexibility in
handling partial rankings of subsets of entities and top-$m$
rankings~\citep{gormley2009grade,mollica2017bayesian,HendersonK18}.
Efficient computational algorithms for fitting Plackett-Luce models
are also well-established \citep{hunter2004mm,GuiverS09,CaronD12}.
However, like many models for rank ordered data, the Plackett-Luce
model cannot explicitly handle ties.

This paper considers a model which can be considered as a
generalization of the Plackett-Luce model that can explicitly handle
ties.  The intuition behind the model is based on the latent variable
representation of the Plackett-Luce model. Under the Plackett-Luce
model, the probability of a particular ordering is equivalent to the
probability that realised values from independent exponential latent 
variables take that particular ordering. It is clearly impossible to
generate ties under the Plackett-Luce model due to the continuous
nature of the latent variables. Replacing the exponential latent
variables by their discrete counterparts, geometric latent variables,
leads to a model which naturally accommodates ties and which
inherits some of the desirable properties of the Plackett-Luce model,
such as tractability and the ability to deal naturally with partial
rankings. The Plackett-Luce model is a limiting special case of this
proposed model and so it can be considered as a generalized
Plackett-Luce model. In the remainder of the paper this model shall be
referred to as the GPL model for brevity, with the G standing for
\textit{Generalized} or \textit{Geometric}. 

The geometric latent variable model is not new, having been described
independently in \cite{BakerS21}. However, this paper makes several
contributions.  Firstly, an explicit form for the likelihood is
derived which facilitates maximum likelihood and Bayesian inference,
and secondly, the model is demonstrated in a range of examples, not
just those involving paired comparisons. The focus in this paper is on
Bayesian inference via the latent variable formulation and as such the
methods rely on those introduced by \cite{CaronD12} for the
Plackett-Luce model. This leads to the derivation of efficient Gibbs
samplers for posterior sampling and EM algorithms for determining
maximum a posterior (MAP) estimators, which is a third significant
contribution of the paper. 

Explicit probabilistic models for rank ordered data with ties are rare. The
Davidson-Luce (DL) model \citep{FirthKT19arxiv,turner2020modelling} is one such
model. The DL model is based on the Davidson model \citep{Davidson70}
which is a generalisation of the Bradley-Terry model \citep{Zermelo29,
BradleyT52} for paired comparisons to account for ties. The
Davidson-Luce model generalises the Davidson model to the case of
rank ordered data (with more than two entities), and inherits many of the
attractive properties of the Davidson model. As a direct competitor
to the GPL model of this paper, the DL model will be discussed in
more detail in Section~\ref{sec:Davidson-Luce}, and some comparisons
made.  \cite{barney2015joint} describe an alternative approach
for explicitly handling ties in rankings based on continuous latent
variables, with the probability of a tie depending on the closeness of
the latent variables.
There is also a related strand of research in the computer science/machine
learning literature; see for example, \cite{truyen2011probabilistic} who introduce a
Probabilistic Model on Ordered Partitions (POMP) which has some
similarities to the Davidson-Luce model.  Considering a rank ordering
with ties as a set of ordered partitions of the entities is appealing;
the partitions are sometimes referred to as ``buckets'' and this leads
to the \textit{bucket order} terminology. The related problem of finding
an optimal bucket order from a set of preferences is the subject of
several papers \citep{gionis2006algorithms,feng2008discovering,kenkre2011discovering}  but is not the main objective of the work presented here.

Given the relative rarity of explicit models for rank ordered data
with ties most analyses rely on a model which does not allow ties
(such as the PL model) together with some form of approximation; see,
for example, \cite{baker2015deterministic} for an overview of
possible approximations.  Although appealing for inferential purposes,
a drawback of these approximate approaches is that the inherent
continuous nature of the underlying models means that they are not
well suited to prediction, since ties cannot be generated and the
probability of a tie in a future comparison cannot be computed. Another
approach suitable for a Bayesian analysis is to randomly break ties
within an Markov chain Monte Carlo (MCMC) algorithm \citep{glickman2015stochastic}. An additional
drawback of this approach is that the sampler must explore the space
of permutations of the entities in each ordered partition and this may
be prohibitively large with large partitions and many ties.

In this paper it will be demonstrated that the GPL model overcomes
several of the drawbacks of existing approaches. For example, unlike
the methods which rely on approximations based on models without ties,
the GPL model is exact, the predictive probability of a tie can be computed and its generative nature is ideally suited to
predictive inference.  Also, unlike some of the aforementioned
explicit models for ties, the GPL model is relatively tractable and
scales up to reasonably sized datasets. The tractability is important
as it admits a closed form likelihood function for the GPL model which
should facilitate computation in widely available software packages,
smoothing the way for developments and extensions.

The remainder of the paper is structured as follows.
Section~\ref{sec:data} describes rank ordered data with ties and
establishes some notation. This is done through a simple illustrative
example which runs throughout the paper. The GPL model is described in
Section~\ref{sec:model} and inference is considered in
Section~\ref{sec:inference}. Three real data sets are used for illustration
in Section~\ref{sec:examples} and the paper concludes in
Section~\ref{sec:discuss} with a discussion.

\section{Data and notation}
\label{sec:data}

Suppose there are $K$ entities represented by the set
$\mathcal{K}=\{1,2,\ldots,K\}$ and $n$ rank orderings of possible
subsets of these entities. It is assumed that the rank orderings have
arisen from a set of comparisons involving two or more of the
entities. These comparisons might be based on preference or
competition.

\subsection{Simple illustrative example}
\label{sec:simp-data}

To fix notation and ideas it is useful to consider the following
simple example. Suppose there are $K=5$ entities labelled
$\mathcal{K}=\{1,2,3,4,5\}$ and suppose they are ranked such that
entity 1 was first, entity 2 was second, then entities 4 and 5 were
tied for third, with entity 3 ranked in last place. A rank ordering
with ties is technically a \textit{weak ordering} of the set of
entities, or a totally ordered partition of the entities, succinctly
referred to as a bucket order. There are several ways to represent
such data; for example, $\{1\},\{2\},\{4,5\},\{3\}$. In this paper
these data are represented as an ordered list of entities
$\vec{y}=(1,2,4,5,3)'$ and an associated ordered set indicator
$\vec{s}=(1,2,3,3,4)'$, such that $s_j$ denotes the ordered set (or
``bucket'') which contains entity $y_j$; $\vec{s}$ defines an ordered
partition of the entities. Here $\vec{s}$ keeps track of the ties; if
$s_{j}=s_{j+1}$ then entities $y_j$ and $y_{j+1}$ are tied, and
therefore if swapped would convey the same information. In other
words, the data in this simple illustrative example could equally have
been stored as $\vec{y}=(1,2,5,4,3)'$, with the same $\vec{s}$.

\subsubsection{Patterns of ties and number of bucket orders}
\label{sec:patterns}

It is also useful to define $t_{j}=\mathbb{I}(s_{j}=s_{j+1})$ for
$j=1,\ldots,K-1$, where $\mathbb{I}(A)$ is the indicator function
which equals 1 if $A$ is true and equals 0 otherwise. Clearly
$t_{j}=1$ if and only if the entities $y_{j}$ and $y_{j+1}$ are tied.
In the simple illustrative example, $\vec{t}=(0,0,1,0)'$, indicating
that entities 4 and 5 are tied. Since $t_i$ can take the value 0 or 1
there are $2^{K-1}$ possible patterns of ties. Therefore, in this
simple example with $K=5$ there are 16 possible tie patterns. The
number of bucket orders of $K$ entities is given by the $K$th ordered
Bell number (or Fubini number); for example, with $K=5$ there are 541
possible bucket orders. The ordered Bell numbers grow quickly, with,
for example, over 102 million possible bucket orders when there are 10
entities. In contrast, if ties are not allowed then there are $K!$
possible orderings of $K$ entities, and, for example, when $K=10$ this
is approximately 3.6 million, far fewer than the number of bucket
orders. 

\subsection{Types of rank orderings}
\label{sec:types-data}

Suppose $M\leq K$ entities are considered in a comparison
then a \textit{complete} rank ordering refers to a rank ordering of
all $M$ entities. Alternatively, if only the top $m< M$ entities are
reported then this is referred to as a top-$m$ rank ordering. If $M<K$
then these complete and top-$m$ rank orderings are based on a subset of
entities. A paired comparison arises when $M=2$.  The GPL model, as
described in Section~\ref{sec:model}, can handle all the above types of
data.

\section{The GPL model}
\label{sec:model}

In this section the proposed GPL model will be
described. First, however, it will be useful to provide a brief
description of the Plackett-Luce model.   

\subsection{The Plackett-Luce model}
\label{sec:pl}

The Plackett-Luce (PL) model provides a probabilistic model for a rank
ordering $\vec{y}$, where $\vec{y}$ is as defined in
Section~\ref{sec:data}. Let $\lambda_k>0$ denote the ``strength''
parameter  (in terminology originating from \cite{Zermelo29}) 
corresponding to entity $k$, for $k\in\mathcal{K}$, and let
$\vec{\lambda}=(\lambda_1,\ldots,\lambda_K)'$. Then under the PL
model, the probability of observing $\vec{y}$ is 
\begin{equation}
\Pr(\vec{Y}=\vec{y}|\vec{\lambda})=  \prod_{j=1}^{K-1}
\frac{\lambda_{y_j}}{\sum_{\ell=j}^{K} \lambda_{y_{\ell}}}. 
\label{eq:pl}
\end{equation}
The probability in~\eqref{eq:pl} is invariant to scalar
multiplication of $\vec{\lambda}$. Therefore only $K-1$ of the
parameters are identifiable from the data when considering
likelihood-based inference on multiple datasets. 

The PL model has a generative representation in terms of ordering
independent exponential latent variables. This latent variable
representation is sometimes referred to as a Thurstonian
representation (for example, \cite{bockenholt1992thurstonian}) after
\cite{Thurstone27} or a random utility model
\citep{train2009discrete}. The latent variables can be thought of as
representing the unobserved utility (for preference data) or
performance/score (for data arising from competition). Specifically,
for the PL model, let $W_k\overset{\text{indep.}}{\sim}
Exp(\lambda_k),\ k\in\mathcal{K}$ be independent exponential random
variables, one for each entity, with means $\lambda_k^{-1}$. Then, it
can be shown that
\begin{equation}
\Pr(W_{y_1} < W_{y_2} < \cdots < W_{y_K}|\vec{\lambda})=\prod_{j=1}^{K-1}
\frac{\lambda_{y_j}}{\sum_{\ell=j}^{K} \lambda_{y_{\ell}}},
\label{eq:pl-latent}
\end{equation}
which is precisely the PL probability in~\eqref{eq:pl}.  It is clear
that~\eqref{eq:pl-latent} holds for the $W_k$ being any
monotonically increasing function of an
exponential random variable. Moreover, these are the only random
variables that lead to a concise closed form expression for
$\Pr(W_{y_1} < W_{y_2} < \cdots < W_{y_K})$; see
\cite{Baker20} and \cite{BakerS21} for further details.  This
tractability adds to the appeal of
the PL model. The PL model also provides a closed form expression for
the probability of a top-$m$ rank ordering, for example, 
\begin{equation}
\Pr(W_{y_1}  < \cdots < W_{y_m} < \{W_{y_{m+1}},\cdots,W_{y_{K}}\}|\vec{\lambda})=\prod_{j=1}^{m}
\frac{\lambda_{y_j}}{\sum_{\ell=j}^{K} \lambda_{y_{\ell}}}.
\label{eq:pl-latent-top}
\end{equation}
The latent exponential random variables are continuous and so as a
generative model for rank ordered data the PL model will never
generate ties, since no two realised values will be identical. The
consequence is that the PL model, like most probabilistic models for
rank ordered data, cannot model ties directly. 

\subsection{The GPL model: basic idea and origins}
\label{sec:basic}

The basic intuition underpinning the GPL model is to replace the
independent exponential random variables in the generative model by
their discrete counterparts, geometric random variables. That is
replace $W_k\overset{\text{indep.}}{\sim} Exp(\lambda_k)$ by
$W_{k}\overset{\text{indep.}}{\sim} \text{Geom}(\theta_k)$ for
$k\in\mathcal{K}$, where $0<\theta_k\leq 1$ represents the strength
 parameter for entity $k$.  Here the random variable $W_k$ has probability mass function
\[
\Pr(W_k=w)=(1-\theta_k)^{w-1}\theta_k, \quad w\in\{1,2,3,\ldots\}.
\]

This representation clearly allows for ties
due to the discrete nature of the geometric latent variables. It will
also be shown that this representation inherits some of the desirable
properties of the PL model such as tractable closed forms for the
probability of rank orderings (both complete and top-$m$) with ties.

\subsubsection{Interpretation in terms of Bernoulli trials}
\label{sec:bern}

The generative process can be reduced to a series of independent Bernoulli trials, one for each entity, such that
entity $k$ has probability of success $\theta_k$ in any one trial,
that is $I_{k \ell}\overset{\text{indep.}}{\sim}
\text{Bern}(\theta_k)$, for $k\in\mathcal{K}$ and $\ell=1,2,\ldots$.
The entities are to be ordered on the number of trials until their
first success, $W_k$, for $k\in\mathcal{K}$. This can be thought of
as the number of flips until a head is flipped of a biased coin
with probability of heads $\theta_k$, with the entities ranked in
terms of the number of flips until flipping a head, from fewest to
most.  Hence, the larger $\theta_k$ is the more likely entity $k$ is
to flip a head first, and therefore be ranked first.

\subsubsection{Comparison to previous work}
\label{sec:BakerS21}

The idea of using geometric latent variables to induce ties was
described independently in \cite{BakerS21} using different notation
and terminology. The main difference to this work is in the definition
of the geometric distribution; \cite{BakerS21} use the number of
failures before the first success and thus the sample space is
$\{0,1,2,\ldots\}$. The two definitions lead to the same model and
likelihood due to the fact that one definition defines random
variables that are linear transformations of those from the other
definition. As pointed out in \cite{BakerS21}, the geometric
distribution (in its two definitions) is the only discrete
distribution for which the following derivations hold. The choice of
$\{1,2,\ldots\}$ for the sample space of the geometric latent
variables in this paper was made independently of knowledge of
\cite{BakerS21}, but nevertheless, such a parameterisation leads to a
neater form for the survival function and hence slightly more
tractable derivations. It also leads to parameters that are slightly
easier to interpret, although in both cases the differences are small.
\cite{BakerS21} focused mainly on the properties of the model and
inference in the case of paired comparisons. They provided an
algorithm for computing the likelihood in the general case rather than
a closed form expression. In this paper such an expression is provided
and the focus is more on inference in the general case.

\subsection{Complete rankings on subsets of entities}
\label{sec:complete}

Complete rank orderings are conceptually simpler than top-$m$ rank
orderings and so they provide the starting point for a description of
the model and the notation used  in the paper. 

\subsubsection{Illustrative example}
\label{sec:complete-simple}

The probability of
observing $\vec{y}=(1,2,4,5,3)',\vec{s}=(1,2,3,3,4)'$ given
$\vec{\theta}=(\theta_1,\ldots,\theta_5)'$ is defined to be the
probability of the latent geometric random variables
$\vec{W}=(W_1,\ldots,W_5)'$ taking the ordering $\vec{y}$ with ties
defined by~$\vec{s}$. The tractability of the geometric distribution,
in particular its survival function, admits a closed form expression
for the above probability,
\begin{align}
  \Pr(&\vec{Y}=(1,2,4,5,3)',\ \vec{S}=(1,2,3,3,4)'\,|\,\vec{\theta})\notag\\
  &=\Pr(W_{1} < W_{2} < W_{4} =  W_{5} < W_{3}\,|\,\vec{\theta})\notag\\
 &=\sum_{w1=1}^\infty \Pr(W_1=w_1) \sum_{w_2=w_1+1}^\infty
 \Pr(W_2=w_2) \sum_{w_4=w_2+1}^\infty \Pr(W_4=w_4)\Pr(W_5=w_4)\sum_{w_3=w_4+1}^\infty\Pr(W_3=w_3)\notag\\
&\begin{aligned}
 &=\frac{\theta_1(1-\theta_2)(1-\theta_4)(1-\theta_5)(1-\theta_3)}{1-(1-\theta_1)(1-\theta_2)(1-\theta_4)(1-\theta_5)(1-\theta_3)}  \times\frac{\theta_2(1-\theta_4)(1-\theta_5)(1-\theta_3)}{1-(1-\theta_2)(1-\theta_4)(1-\theta_5)(1-\theta_3)}\\
&\qquad\qquad\qquad\qquad\times
\frac{\theta_4\theta_5(1-\theta_3)}{1-(1-\theta_4)(1-\theta_5)(1-\theta_3)}
,
\end{aligned}
\label{eq:simple}
 \end{align}
where dependence on $\vec{\theta}$ has been suppressed on the third
line for brevity. It is helpful to note that due to the tractability
of the geometric distribution  this probability can be factorised into the product of probabilities
 involving pairs of random variables, that is
\begin{align*}
\Pr(W_{1} < &W_{2} < W_{4} =  W_{5} < W_{3}\,|\,\vec{\theta})\\
  &=\Pr(W_{1}< \min \{W_2,W_3,W_4,W_5\}\,|\,\vec{\theta})\Pr(W_{2}< \min
  \{W_3,W_4,W_5\}\,|\,\vec{\theta})\\
& \qquad\times \Pr(W_{4}= \min
  \{W_3,W_5\}\,|\,\vec{\theta})\Pr(W_5 < W_3\,|\,\vec{\theta})\\
 &=\prod_{j=1}^4 \Pr(W_{y_j}< \min_{\ell>j}
  W_{y_\ell}\,|\,\vec{\theta})^{1-\mathbb{I}(s_{j}=s_{j+1})} 
\Pr(W_{y_j}=\min_{\ell>j}
  W_{y_\ell}\,|\,\vec{\theta})^{\mathbb{I}(s_{j}=s_{j+1})}.
\end{align*}
This factorisation suggests an alternative, multi-stage \citep{FlignerV88},
representation of the GPL model.   In the above example, in the first
stage, entity 1 is
chosen from the set of available entities
$\mathcal{R}_1=\{1,2,3,4,5\}$. This has probability
$\Pr(W_{1}< \min \{W_2,W_3,W_4,W_5\}\,|\,\vec{\theta})$. Entity 1 is
removed leaving the set of available entities at the second stage as
$\mathcal{R}_2=\{2,3,4,5\}$. At this second stage, entity 2
is chosen from the set $\mathcal{R}_2$, and this has probability
$\Pr(W_{2}< \min \{W_3,W_4,W_5\}\,|\,\vec{\theta})$. At the third
stage, the set of available entities is $\mathcal{R}_3=\{3,4,5\}$ and
entities 4 and 5 are chosen (that is, the bucket $\{4,5\}$) from
$\mathcal{R}_3$. This probability can be factorised as $\Pr(W_{4}= \min
  \{W_3,W_5\}\,|\,\vec{\theta})\Pr(W_5 < W_3\,|\,\vec{\theta})$. 

These paired comparison probabilities, $\Pr(X<Y)$ and $\Pr(X=Y)$, each have a simple closed form
(see Appendix~\ref{app:geom-prop})  giving in this example  
\begin{align}
\Pr(W_{1} < &W_{2} < W_{4} = W_{5} < W_{3}\,|\,\vec{\theta})\notag\\
&=\Pr(W_{1}< \min \{W_2,W_3,W_4,W_5\}\,|\,\vec{\theta})\Pr(W_{2}< \min
  \{W_3,W_4,W_5\}\,|\,\vec{\theta})\notag\\
& \qquad\times\Pr(W_{4}= \min
  \{W_3,W_5\}\,|\,\vec{\theta})\Pr(W_5 < W_3\,|\,\vec{\theta})\notag\\
&=\frac{\theta_1(1-\theta_2)(1-\theta_4)(1-\theta_5)(1-\theta_3)}{1-(1-\theta_1)(1-\theta_2)(1-\theta_4)(1-\theta_5)(1-\theta_3)}\times\frac{\theta_2(1-\theta_4)(1-\theta_5)(1-\theta_3)}{1-(1-\theta_2)(1-\theta_4)(1-\theta_5)(1-\theta_3)}\notag\\
&\qquad\qquad\qquad\qquad\times
\frac{\theta_4(1-(1-\theta_5)(1-\theta_3))}{1-(1-\theta_4)(1-\theta_5)(1-\theta_3)}\times
\frac{\theta_5(1-\theta_3)}{1-(1-\theta_5)(1-\theta_3)}\notag\\
&\begin{aligned}
&=\frac{\theta_1(1-\theta_2)(1-\theta_4)(1-\theta_5)(1-\theta_3)}{1-(1-\theta_1)(1-\theta_2)(1-\theta_4)(1-\theta_5)(1-\theta_3)}\times\frac{\theta_2(1-\theta_4)(1-\theta_5)(1-\theta_3)}{1-(1-\theta_2)(1-\theta_4)(1-\theta_5)(1-\theta_3)}\\
&\qquad\qquad\qquad\qquad\times
\frac{\theta_4\theta_5(1-\theta_3)}{1-(1-\theta_4)(1-\theta_5)(1-\theta_3)},
\end{aligned}
\label{eq:stagewise}
\end{align}
 which matches up with Equation~\eqref{eq:simple}.  The stage-wise
 probabilities in the factorisation of the joint probability
 in~\eqref{eq:stagewise} take a
 relatively simple form, and are directly interpretable in terms of the
 Bernoulli random variables constituting the geometric latent
 variables. Each stage-specific probability represents the probability
 that the chosen entities were successful (hypothetically flipped a
 head) and the un-chosen entities were unsuccessful (hypothetically
 flipped a tail), given that that there is at least one success. In
 other words, the numerator in each term is the product of the
 probabilities of success $(\theta)$ for the chosen entities multiplied
 by the product of the probabilities of failure $(1-\theta)$ for the
 un-selected entities; the denominator is 1 minus the probability all
 available entities were unsuccessful.

\subsubsection{General case}
\label{sec:complete-general}

Now consider the general case in which data on $n$ orderings involving
$K$ entities $\mathcal{K}=\{1,2,\ldots,K\}$ are observed.  Suppose
that $n_i$ entities $\mathcal{K}_i\subseteq \mathcal{K}$ are
considered in ordering $i$. Data for the $i$th ordering are
$\vec{y}_i=(y_{i1},y_{i2},\ldots,y_{in_i})'$ in which $y_{ij}$ denotes
the entity listed in $j$th position in the $i$th ordering and
$\vec{s}_i=(s_{i1},s_{i2},\ldots,s_{in_i})'$ in which $s_{ij}$ denotes
the ordered set which contains entity $y_{ij}$. For each entity $k$,
associate a parameter $0<\theta_k\leq 1$. Let
$W_{ik}\overset{\text{indep.}}{\sim} \text{Geom}(\theta_k)$ for
$i=1,\ldots,n$ and $k\in\mathcal{K}_i$. Then,
\begin{align*}
\Pr(\vec{Y}_i=\vec{y}_i,\vec{S}_i=\vec{s}_i|\vec{\theta})
&= \prod_{j=1}^{n_i-1} \Pr(W_{iy_{ij}} < \min_{\ell>j} W_{iy_{i\ell}}|\vec{\theta})^{1-\mathbb{I}(s_{ij}=s_{i,j+1})}\Pr(W_{iy_{ij}} = \min_{\ell>j} W_{iy_{i\ell}}|\vec{\theta})^{\mathbb{I}(s_{ij}=s_{i,j+1})}\\
&= \prod_{j=1}^{n_i-1} \left\{ \frac{\theta_{y_{ij}}\prod_{\ell>j} (1-\theta_{y_{i\ell}}) }{1-(1-\theta_{y_{ij}})\prod_{\ell>j}(1-\theta_{y_{i\ell}})} \right\}^{1-t_{ij}} \left[ \frac{\theta_{y_{ij}}\left\{1-\prod_{\ell>j} (1-\theta_{y_{i\ell}})\right\} }{1-(1-\theta_{y_{ij}})\prod_{\ell>j}(1-\theta_{y_{i\ell}})} \right]^{t_{ij}}\\
&=\prod_{j=1}^{n_i-1} \frac{\theta_{y_{ij}}\left\{ \prod_{\ell>j}(1-\theta_{y_{i\ell}})\right\}^{1-t_{ij}}\left\{ 1-\prod_{\ell>j}(1-\theta_{y_{i\ell}})\right\}^{t_{ij}}}{1-\prod_{\ell\geq j}(1-\theta_{y_{i\ell}})},
\end{align*}
where $t_{ij}=\mathbb{I}(s_{ij}=s_{i,j+1})$ for $i=1,\ldots,n$ and
$j=1,\ldots,n_i-1$.  

\subsubsection{Likelihood}
\label{sec:like-complete}

Suppose the data consists of $n$ independent rank orderings. Let
$D=\{\vec{y}_i,\vec{s}_i\}_{i=1}^n$ denote the observed data.  
The likelihood function for $\vec{\theta}$ is
\begin{align}
L(\vec{\theta}|D)\equiv p(D|\vec{\theta}) & = \prod_{i=1}^n \Pr(\vec{Y}_i=\vec{y}_i,\vec{S}_i=\vec{s}_i|\vec{\theta})\notag\\
& = \prod_{i=1}^n \prod_{j=1}^{n_i-1} \frac{\theta_{y_{ij}}\left\{ \prod_{\ell>j}(1-\theta_{y_{i\ell}})\right\}^{1-t_{ij}}\left\{ 1-\prod_{\ell>j}(1-\theta_{y_{i\ell}})\right\}^{t_{ij}}}{1-\prod_{\ell\geq j}(1-\theta_{y_{im}})}.
\label{eq:like}
\end{align}
Note that further simplification of the likelihood is possible but this is deferred to
Section~\ref{sec:like-top} where the more general case of top-$m$
ranking orderings on subsets of entities is considered.

An additional attractive feature of the GPL model is that there are no
issues with parameter identifiability, unlike the PL model; all $K$
parameters are likelihood identifiable. 

\subsubsection{Reversing the order of the latent variables}
\label{sec:reverse}

In the GPL model as described so far the latent geometric random
variables are ordered from the smallest to the largest. This can be
referred to as the \textit{smaller is better} version of the GPL.  The
alternative, in which the geometric latent variables are ordered from
the largest to the smallest, can also be considered. This can be
referred to as the \textit{bigger is better} version of the GPL. This
version of the model might be appropriate for modelling sports where
the latent variables represent a performance or score and the larger
the score the better.

The bigger is better version of the GPL model can be fitted by simply
reversing the ordering in $\vec{y}$
($\vec{y}_i^{r}=\text{rev}(\vec{y}_i)$ for $i = 1,\ldots,n$, where
$\text{rev}(\vec{x})$ reverses the order of the elements in the vector
$\vec{x}$), making the appropriate change to $\vec{s}$
($\vec{s}_i^r=s_{i n_i}-\text{rev}(\vec{s}_i)$ for $i = 1,\ldots,n$),
and then analysing $D=\{\vec{y}_i^r,\vec{s}_i^r\}_{i=1}^n$ under the
usual (smaller is better) GPL model.  For example, if
$\vec{y}=(1,2,4,5,3)'$ and $\vec{s}=(1,2,3,3,4)'$ then  
$\vec{y}^{r}=(3,5,4,2,1)'$ and $\vec{s}^r=(1,2,2,3,4)'$.  The
interpretation of the parameters changes, however. Now small values of
$\theta$ are good. Due to the simple act of reversing the data the
bigger is better version of the GPL model shall be referred to
henceforth as the \textit{reverse GPL} model.

\subsection{Top rankings on subsets of entities: the most general case}
\label{sec:top}

Suppose in comparison $i$, the top $m_i$ entities out of the $n_i$
entities in $\mathcal{K}_i\subseteq \mathcal{K}$ are reported. For
example, with preference data it is often easier to give a rank
ordering for the top $m<n_i$ entities than it is to give a full ranking
of all $n_i\leq K$ entities, especially when the number of entities is
large. Similarly, in a sporting context with multiple participants,
interest and therefore reported results may focus on the top
positions rather than the complete ranking of all
competitors. Fortunately, the GPL model can naturally handle such
scenarios. Note that this is the most general scenario and contains
complete rank orderings as a special case.

Let $\vec{y}_i$ be such that $y_{i1},\ldots,y_{im_i}$ contain
the top $m_i$ entities in order, and the remaining $n_i-m_i$ elements
in $\vec{y}_i$ are the un-ranked entities in arbitrary order.
Similarly, let $\vec{s}_i$ be such that $s_{i1},\ldots,s_{im_i}$
contain the ordered set indicators for the top $m_i$ entities, and the
remaining $n_i-m_i$ elements in $\vec{s}_i$ are assigned to the
hypothetical $s_{i m_i}+1$th set. Assigning all un-ranked entities to
the same set is purely for notational convenience; if they were to be
ranked then they need not be restricted to being tied. The key point
is that the set indicators $\vec{s}_i$ identify all entities not in
the top $m_i$.

Another scenario might be that the $n_i$ entities have been ranked but
retrospectively a top $q_i<n_i$ is required. In this case the ordered
set containing the $q_i$th ordered entity is found and then all
entities in that set and the higher ranked sets would be included in
the data, that is $m_i=\sum_{j=1}^{n_i} \mathbb{I}(s_{ij}\leq
s_{iq_i})$. For example, consider the running example with
$\vec{y}=(1,2,4,5,3)'$ and $\vec{s}=(1,2,3,3,4)'$ and suppose the top
$q=3$ entities are required. The $q=3$rd ranked entity, $4$, is in the
third ranked set, so all entities in the third and
higher ranked sets must be considered;  this corresponds to  the
entities $1,2,4,5$, and therefore this would
be a top $m=4$ rank ordering.

As a concrete example with the same data, suppose that $q=2$, then $m=2$
as the second ranked entity is in the second set and only two entities
are ranked in the top two sets. Therefore, the probability of
observing $\vec{y},\vec{s}$ given
$\vec{\theta}=(\theta_1,\ldots,\theta_K)'$ and $m=2$ is defined to be
the probability of the latent geometric random variables
$\vec{W}=(W_1,\ldots,W_K)'$ taking the ordering $\vec{y}$ with ties
defined by~$\vec{s}$, up to the entity ranked in $m$th position. That
is,
\begin{align*}
  \Pr(\vec{Y}=(1,2,4,5,3)',&\ \vec{S}=(1,2,3,3,4)'\,|\,\vec{\theta},m=2)\\
  &=\Pr(W_{1} < W_{2} < \left\{W_{3},W_{4},W_{5}\right\}\,|\,\vec{\theta})\\
  &=\Pr(W_{1}< \min \{W_2,W_3,W_4,W_5\}\,|\,\vec{\theta})\Pr(W_{2}< \min
  \{W_3,W_4,W_5\}\,|\,\vec{\theta})\\
 &=\prod_{j=1}^m \Pr(W_{y_j}< \min_{\ell>j}
  W_{y_\ell}\,|\,\vec{\theta})^{1-\mathbb{I}(s_{j}=s_{j+1})} 
\Pr(W_{y_j}=\min_{\ell>j}
  W_{y_\ell}\,|\,\vec{\theta})^{\mathbb{I}(s_{j}=s_{j+1})}.
\end{align*}

Therefore, in general, for $1\leq m_i \leq n_i$ 
\begin{align*}
\Pr(\vec{Y}_i&=\vec{y}_i,\vec{S}_i=\vec{s}_i|\vec{\theta},m_i)\\
&= \prod_{j=1}^{m_i^\star} \Pr(W_{iy_{ij}} < \min\left\{W_{iy_{i\ell}}\right\}_{\ell=j+1}^{n_i}|\vec{\theta})^{1-\mathbb{I}(s_{ij}=s_{i,j+1})}\Pr(W_{iy_{ij}} = \min\left\{W_{iy_{i\ell}}\right\}_{\ell=j+1}^{n_i}|\vec{\theta})^{\mathbb{I}(s_{ij}=s_{i,j+1})}\\
&=\prod_{j=1}^{m_i^\star} \frac{\theta_{y_{ij}}\left\{ \prod_{\ell=j+1}^{n_i}(1-\theta_{y_{i\ell}})\right\}^{1-t_{ij}}\left\{ 1-\prod_{\ell=j+1}^{n_i}(1-\theta_{y_{i\ell}})\right\}^{t_{ij}}}{1-\prod_{\ell=j}^{n_i}(1-\theta_{y_{i\ell}})},
\end{align*}
where $m_i^\star=\min(m_i,n_i-1)$.  Put simply, the usual product of
terms is truncated at $m_i^\star$.

\subsubsection{Likelihood}
\label{sec:like-top}

Let $\vec{m}=(m_1,\ldots,m_n)'$ denote the collection of truncation
levels of each of the $n$ orderings. Assuming that these rank
orderings are independent, the likelihood function for $\vec{\theta}$ is
\begin{align}
L(\vec{\theta}|D,\vec{m})\equiv p(D|\vec{\theta},\vec{m}) & = \prod_{i=1}^n \Pr(\vec{Y}_i=\vec{y}_i,\vec{S}_i=\vec{s}_i|\vec{\theta},m_i)\notag\\
& = \prod_{i=1}^n \prod_{j=1}^{m_i^\star} \frac{\theta_{y_{ij}}\left\{ \prod_{\ell=j+1}^{n_i}(1-\theta_{y_{i\ell}})\right\}^{1-t_{ij}}\left\{ 1-\prod_{\ell=j+1}^{n_i}(1-\theta_{y_{i\ell}})\right\}^{t_{ij}}}{1-\prod_{\ell=j}^{n_i}(1-\theta_{y_{i\ell}})}\notag\\
&=\frac{\prod_{k=1}^K \theta_k^{w_k}(1-\theta_k)^{d_k}}{\prod_{i=1}^n
  \prod_{j=1}^{v_i}\left\{
    1-\prod_{\ell\in\mathcal{R}_{ij}}(1-\theta_{\ell})\right\}}
\label{eq:like-top}
\end{align}
where $v_i=s_{i m_i^\star}$ denotes the number of
ordered sets containing the top $m_i$ entities but not including the
$s_{in_i}$th ordered set if it does not involve a
tie. Let $\delta_{ijk}=\mathbb{I}(k \in
\mathcal{R}_{ij})$ for $k\in\mathcal{K}, i=1,\ldots,n,
j=1,2,\ldots,s_{i n_i}$, be an indicator variable which 
is equal to 1 if entity $k$ is ranked at least as low as the
$j$th set in the $i$th comparison, and is equal to 0 otherwise. 
Then    $d_k= \sum_{i=1}^n\sum_{j=2}^{s_{in_i}}\delta_{ijk}$ denotes
the number of sets ranked higher than the set that contains entity
$k$ (alternatively, $d_k= \sum_{i=1}^n
\sum_{j=1}^{n_i}\left\{(s_{ij}-1)\mathbb{I}(y_{ij}=k)\right\}$);  
$w_k =  \sum_{i=1}^n\sum_{j=1}^{v_i}\delta_{ijk} -d_k$ 
denotes the number of comparisons in which entity $k$ appears in the
top $m_i$ (and is not in position $n_i$ on
its own). Note also that $c_k=\sum_{i=1}^n \sum_{j=1}^{n_i}
\mathbb{I}(y_{ij}=k)$ denotes the number of comparisons involving
entity $k$. In the simple illustrative example, $m=K$, 
$\vec{w}=(w_1,\ldots,w_K)'=(1,1,0,1,1)'$,
$\vec{c}=(c_1,\ldots,c_K)'=(1,1,1,1,1)'$,
$\vec{d}=(d_1,\ldots,d_K)'=(0,1,3,2,2)'$,
$\mathcal{R}_1=\left\{1,2,3,4,5\right\}$,
$\mathcal{R}_2=\left\{2,3,4,5\right\}$,
$\mathcal{R}_3=\left\{3,4,5\right\}$, $\mathcal{R}_4=\left\{3\right\}$,
and $v=3$. 

\subsubsection{Special cases}
\label{sec:special-top}

The GPL model is capable of dealing with the general case of top-$m$
rank orderings on subsets of entities. However, such data contain
several special cases. For example, if $m_i=n_i$ then it is a complete
rank ordering on subsets of entities and the likelihood contribution
is equivalent to that described in Section~\ref{sec:like-complete} and
given in \eqref{eq:like}. If $m_i=n_i=2$ then it is a paired
comparison.

\subsection{Properties and relationship to other models}
\label{sec:properties}

\subsubsection{Interpretation of parameters and basic properties}
\label{sec:int-prop}

Consider the standard GPL model. Then $E[W_k]=1/\theta_k$ and so large
values of $\theta_k$ lead to small values of $W_k$ and therefore a
position higher up the rank ordering. In other words, $\theta_k>\theta_i$ implies that
entity $k$ is more likely to be ranked higher than entity $i$. Let
$p_k$ denote the probability that entity $k$ is ranked first on its
own in a comparison involving all $K$ entities, then under the GPL
model  
\begin{equation}
p_k = \frac{\theta_k
  \prod_{i\in\mathcal{K}\setminus{\{k\}}}(1-\theta_i)}{1-\prod_{j\in\mathcal{K}}(1-\theta_j)},
\label{eq:entityk-first}
\end{equation}
where $\mathcal{K}\setminus{\{k\}}$ denotes the set $\mathcal{K}$ with
the $k$th element removed. Clearly this probability will increase as $\theta_k$ increases with
$\theta_i$ remaining fixed for $i\neq k$. 

The $\theta$ parameters also control the probability
of ties, with, in general, ties becoming more prevalent as
$\theta_k\rightarrow 1$. This is because $E[W_k]\rightarrow 1$ as
$\theta_k\rightarrow 1$, and ties are more likely with lower values of
the latent variables. Conversely, taking $\theta_k\rightarrow 0$
implies that ties become less prevalent and in fact the probability of
a tie tends to 0; this is formalised in Section~\ref{sec:PL} where it is
shown that the PL model is a limiting special case of the GPL model. 

This dual interpretation of $\theta$ explains why all $K$ parameters are
identifiable; it is not simply the ratio of the parameters  that impart
information (as in the Plackett-Luce model), but the overall scale of
the $\theta$s is also important. 

Consider the case of paired comparisons between entities $i$ and $j$,
then 
\begin{equation}
\begin{aligned}
\Pr(i \text{ beats }
j)&=\frac{\theta_i(1-\theta_j)}{1-(1-\theta_i)(1-\theta_j)}\\
\Pr(i \text{ ties } j)&=\frac{\theta_i\theta_j}{1-(1-\theta_i)(1-\theta_j)}.
\end{aligned}
\label{eq:gpl-pair}
\end{equation}
Figure~\ref{fig:wintielose} plots the probabilities of a win for $i$ and $j$
and a tie between them for a range of values of $\theta_i\in(0,1)$ when
$\theta_j\in\{0.1,0.5,0.9\}$. 
\begin{figure}[tbh]
\centering
\includegraphics[width=0.9\linewidth]{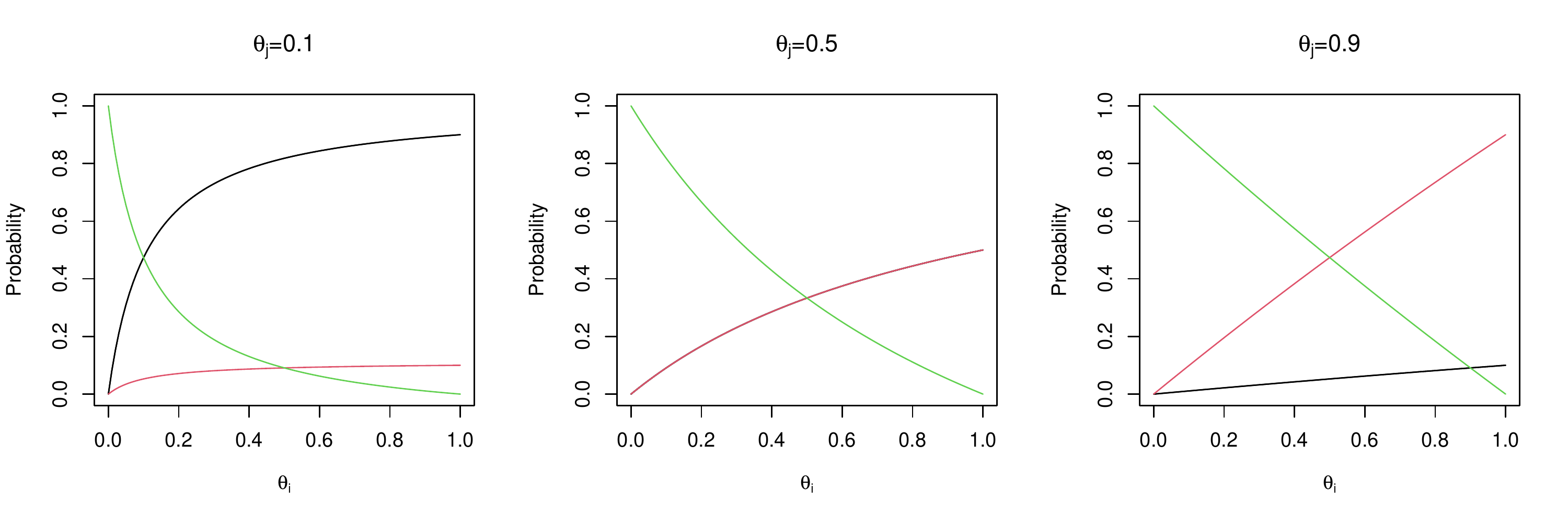} \\
\caption{Probabilities of a win for $i$ (black) a tie (red) or a win
  for $j$ in a paired comparison between entities $i$ and $j$ under
  the GPL model for three different $\theta_j$ values: $0.1$ (left),
  $0.5$ (centre) and $0.9$ (right). In the centre plot the black and
  red lines overlap. }
\label{fig:wintielose}
\end{figure}

Clearly the probability of a win for $i$
or a tie increase as $\theta_i$ increases. For fixed $\theta_i$ the probability of a tie
increases as $\theta_j$ increases.  When the two entities are evenly matched with $\theta_i=\theta_j=\theta$
then $\Pr(i \text{ ties } j)=\frac{\theta}{2-\theta}$.

\subsubsection{Plackett-Luce is a limiting special case}
\label{sec:PL}

It is intuitive to view the Plackett-Luce model as a limiting special
case of the GPL model as the exponential distribution can be seen as
the continuous limiting case of the geometric
distribution. Theorem~\ref{th:pl-limit} formalises this result. 

\begin{thm}
Consider the GPL model with entities in $\mathcal{K}$ and parameters
$\vec{\theta}$. Let $\theta_i=\lambda_i\theta_1$ for $i=1,\ldots,K$ 
such that $\lambda_1=1$ and $\lambda_i>0$ for $i=2,\ldots,K$. For any
ordering without ties $\vec{y}$, $\vec{s}=(1,\ldots,1)'$, under the GPL model 
\[
\lim_{\theta_1\rightarrow 0}\Pr(\vec{y},\vec{s}|\vec{\theta}) = \prod_{i=1}^{K-1} \frac{\lambda_i}{\sum_{j=1}^K\lambda_j},
\]
where the right-hand side is the Plackett-Luce
probability~\eqref{eq:pl}. Hence the PL model is a limiting special case of the GPL model. 

\label{th:pl-limit}
\end{thm}

\begin{proof}
  Consider the
  first stage of the ranking process where the entities with the
  smallest values of the latent geometric random variables
  ``win''. It can be shown that 
\begin{equation}
\lim_{\theta_1\rightarrow 0}\Pr(\text{entity $i$ wins outright}) =
\frac{\lambda_i}{\sum_{j=1}^K \lambda_j}
\label{eq:iwins-pl-limit}
\end{equation}
for all $i\in\mathcal{K}$, and 
\begin{equation}
\lim_{\theta_1\rightarrow 0}\Pr(\text{entities in $\vec{B}$ tie for the win}) = 0
\label{eq:ties-pl-limit}
\end{equation}
for all entities in $\vec{B}\subseteq\mathcal{K}$, where $\vec{B}$
contains more than one entity.
Equation~\eqref{eq:ties-pl-limit} implies the intuitive result that no
ties will be generated by the GPL generative model in the limit as
$\theta_1\rightarrow 0$.  The proof of the above preliminary results
\eqref{eq:iwins-pl-limit} and~\eqref{eq:ties-pl-limit} is elementary and
is based on the strategy
employed in \cite{BakerS21} for the simpler special case of paired
comparisons. Starting with the GPL
probability~\eqref{eq:entityk-first}
\begin{align*}
\Pr(\text{entity $i$ wins outright}) 
& = \frac{\theta_i
                    \prod_{j\neq i}(1-\theta_j)}{1-\prod_{j\in\mathcal{K}}(1-\theta_j)}\\
                  & = \frac{\lambda_i+\theta_1 f_1(\theta_1,\vec{\lambda})}{1+\sum_{j=2}^K\lambda_j +
                    \theta_1 f_2(\theta_1,\vec{\lambda}) }\\
                  & \rightarrow \frac{\lambda_i}{\sum_{j=1}^K\lambda_j},
\end{align*}
as $\theta_1\rightarrow 0$, noting that $\lambda_1=1$ by
definition. Here $f_1(\theta_1,\vec{\lambda})$ and
$f_2(\theta_1,\vec{\lambda})$ represent functions of $\theta_1$ and
$\vec{\lambda}$.
Similarly,
\begin{align*}
\Pr(\text{entities in $\vec{B}$ tie for the win}) & = \frac{\prod_{i \in \vec{B}}\theta_i
                    \prod_{j\notin \vec{B}}(1-\theta_j)}{1-\prod_{j\in\mathcal{K}}(1-\theta_j)}\\
                  & = \frac{\theta_1 f_3(\theta_1,\vec{\lambda})}{1+\sum_{i=2}^K\lambda_i +
                    \theta_1 f_2(\theta_1,\vec{\lambda}) }\\
                  & \rightarrow 0,
\end{align*}
as $\theta_1\rightarrow 0$, where $f_3(\theta_1,\vec{\lambda})$ also
represents a function of $\theta_1$ and
$\vec{\lambda}$.

Now consider the multi-stage representation of the
GPL model and apply the above two preliminary results
\eqref{eq:iwins-pl-limit} and~\eqref{eq:ties-pl-limit} to the probabilities of the
choices of entities at each stage. Taking the limit as
$\theta_1\rightarrow 0$ gives the desired result
\[
\lim_{\theta_1\rightarrow 0}\Pr(\vec{y},\vec{s}|\vec{\theta}) =
\prod_{i=1}^{K-1} \frac{\lambda_i}{\sum_{j=1}^K\lambda_j}.
\]
\end{proof}

A straightforward corollary to Theorem~\ref{th:pl-limit} is that the Bradley-Terry model is a limiting
special case of the GPL model for the case of paired comparisons; see \cite{BakerS21} for further
details. 

\subsubsection{Relationship to the Davidson model}
\label{sec:Davidson}

The Davidson model was introduced in \cite{Davidson70} as a
generalisation of the Bradley-Terry model which can handle ties in
paired comparisons. For a paired comparison between entities $i$ and
$j$ the 
win and tie probabilities under the Davidson model are
\begin{align*}
\Pr(i \text{ beats } j)&=\frac{\lambda_i}{\lambda_i+\lambda_j+\delta\sqrt{\lambda_i\lambda_j}}\\
\Pr(i \text{ ties } j)&=\frac{\delta\sqrt{\lambda_i\lambda_j}}{\lambda_i+\lambda_j+\delta\sqrt{\lambda_i\lambda_j}},
\end{align*}
where $\lambda_k>0$ are the ``worth'' parameters of each entity and
$\delta\geq 0$ controls the probability of a tie. When $\delta=0$ the
Bradley-Terry model is recovered.  For comparison, the corresponding probabilities
under the GPL model (see~\eqref{eq:gpl-pair}) 
can be written as
\begin{align*}
\Pr(i \text{ beats } j)&=\frac{\theta_i(1-\theta_j)}{\theta_i+\theta_j-\theta_i\theta_j}\\
\Pr(i \text{ ties } j)&=\frac{\theta_i\theta_j}{\theta_i+\theta_j-\theta_i\theta_j}.
\end{align*}
In fact, 
\cite{BakerS21} showed that the Davidson model is
equivalent to the reverse GPL model for paired comparisons under
certain restrictions on the parameters. However, in general, the
(reverse) GPL
model for paired comparisons and the Davidson model are different and
will lead to different inferences.  Note that, as in the Plackett-Luce
and Bradley-Terry models, only $K-1$ of the $\lambda$ parameters are
likelihood identifiable. Therefore both the GPL model and the Davidson
model have $K$ identifiable parameters. 

\subsubsection{Relationship to the Davidson-Luce model}
\label{sec:Davidson-Luce}

\cite{FirthKT19arxiv} describe an original model called the
Davidson-Luce (DL) model for choosing the ``winner'' from a set of
entities, where ties are permitted. The construction of the DL model
is based on the Davidson model for paired comparisons with ties. The
model was extended in \cite{turner2020modelling} to rank orderings
with ties by treating it as a multi-stage model with the basic DL
probability making up the stage-wise probabilities.

In
the Davidson-Luce (DL) model, entity $k$ has parameter $\lambda_k>0$ and a
$p$-way 
tie has parameter $\delta_p>0$, where $p$ denotes the
number of entities involved in the tie, and $\delta_1=1$. The Davidson-Luce model is a
multi-stage model with the probability at each stage 
derived as the model for the ``winning'' entities (including ties) in
a comparison involving $K>2$ entities. With $K$ entities there are
$2^{K}-1$ possible outcomes for the ``winner''. For example, when
$K=3$ the $2^3-1=7$ possible outcomes for the ``winner'' can be
denoted $\{1,2,3,(1,2),(1,3),(2,3),(1,2,3)\}$, where $i$ denotes that
entity $i$ was the sole winner, $(i,j)$ denotes that entities $i$ and
$j$ tied for first, and $(i,j,k)$ indicates a three-way tie for first
by entities $i,j,k$.

Generalising the probabilities under the Davidson model to more than
two entities,
the probabilities in the simple example with $K=3$ are 
\begin{align*}
\Pr(\text{$i$ wins})&\propto \lambda_i\\
\Pr(\text{$(i,j)$ win})&\propto \delta_2(\lambda_i\lambda_j)^{1/2}\\
\Pr(\text{$(i,j,k)$ win})&\propto \delta_3(\lambda_i\lambda_j\lambda_k)^{1/3},
\end{align*} 
and the normalising constant for each probability takes the form  
\[
\lambda_i+\lambda_j+\lambda_k+\delta_2(\lambda_i\lambda_j)^{1/2}+\delta_2(\lambda_i\lambda_k)^{1/2}+\delta_2(\lambda_j\lambda_k)^{1/2}+\delta_3(\lambda_i\lambda_j\lambda_k)^{1/3},
\]
the sum of the probabilities for each of the $2^K-1$ outcomes. 

Now consider the general case with $K$ entities.  Let
  $\vec{\omega}=(\omega_1,\ldots,\omega_K)'$ with
  $\omega_i=1$ if entity $i$ is in the winning set and $\omega_i=0$,
  otherwise.  So, for
  example, when $K=3$, the 7 possible outcomes listed above, can be
  collected in the set $\Omega=\{(1,0,0), (0,1,0), (0,0,1), (1,1,0),
  (1,0,1), (0,1,1), (1,1,1)\}$.  The probability of outcome $\vec{\omega}\in\Omega$ under the Davidson-Luce model is
\begin{equation}
\Pr(\vec{\omega})=
\bigfrac{\delta_{\sum_{i=1}^K\omega_i}\left(\prod_{i=1}^K\lambda_i^{\omega_i}\right)^{\frac{1}{\sum_{i=1}^K\omega_i}}}{\sum_{\vec{\omega}'\in\Omega}
\delta_{\sum_{i=1}^K\omega'_i}\left(\prod_{i=1}^K\lambda_i^{\omega'_i}\right)^{\frac{1}{\sum_{i=1}^K\omega'_i}}},
\label{eq:DL}
\end{equation}
where the sum in the denominator involves $|\Omega|=2^{K}-1$
terms. As noted by \cite{turner2020modelling} the sum in the denominator can be truncated at
the maximum observed tie order $p$, but as also noted by \cite{turner2020modelling} the
combinatorial explosion of terms means that data with ties of order
$p>4$ may be prohibitive. 

As the DL model is a multi-stage model the stage-specific probabilities
given by~\eqref{eq:DL} are multiplied together to obtain the joint
probability of a particular rank ordering with ties. 

In general, when there are $K$ entities the DL model has $K+(K-1)$,
that is, the $K$ worth parameters ($\lambda_k$) and the $K-1$ ties
parameters $\delta_p$.  Given that only $K-1$ of the worth parameters
are likelihood identifiable the DL model has $2(K-1)$ parameters.
Hence the DL model has more parameters than the GPL for $K>2$. 

The DL
model has several desirable features, such as intuitive interpretation
 of the parameters,  as described in \cite{FirthKT19arxiv} and
represents a richer, more complex, model than the GPL. However, the
added complexity may mean that the DL model does not scale well to
situations where there are ties involving many entities.
 
\section{Inference based on the GPL model}
\label{sec:inference}

Given some data $D$ it is of interest to make inferences about the
values of the parameters $\vec{\theta}$ of the GPL model. Assume the most general
scenario of top-$m$ rankings on subsets (which contains complete
rankings as a special case) as described in Section~\ref{sec:top}.
Then given the analytical expression for the likelihood
$L(\vec{\theta}|D)$ in~\eqref{eq:like-top} it is straightforward to
maximise the likelihood by using standard numerical methods, such as
those provided by the \texttt{optim} function in \textsf{R} \citep{R}.
Similarly, if adopting a Bayesian approach to inference \citep{BernardoS94} it is straightforward to code the Bayesian model
(including the prior distribution) in a probabilistic programming
language such as \textsf{Stan}~\citep{carpenter2017stan} or \textsf{JAGS}~\citep{plummer2003jags}, and such an approach will be
useful when exploring extensions of the basic GPL.

However, the latent variable formulation of the GPL model suggests
that data augmentation might be useful for inference and this is the
approach taken in this paper. The focus is on Bayesian inference based
on a simple but efficient Gibbs sampler. A simple
Expectation-Maximisation (EM) algorithm for obtaining
MAP and maximum likelihood estimates (MLEs) is also described later in
Section~\ref{sec:em}.  The main advantage of the Gibbs sampler over
other MCMC methods is that it does not
require any tuning.

\subsection{Bayesian inference}
\label{sec:bayes}

\cite{CaronD12} showed how to construct efficient Gibbs samplers and
EM algorithms for the Bradley-Terry model and the Plackett-Luce model.
In what follows, their general approach is used to construct efficient
Gibbs samplers and EM algorithms for the GPL model.

\subsubsection{Latent variables and complete data likelihood}
\label{sec:latent}

For each comparison $i$ introduce $v_i$ latent variables, as follows. For
$i=1,\ldots,n$ and $j=1,\ldots,v_i$ generate
$Z_{ij}|D,\vec{\theta},\vec{m}\overset{\text{indep.}}\sim\text{Geom}(1-\prod_{\ell\in\mathcal{R}_{ij}}(1-\theta_\ell))$,
so that
\[
p(\vec{Z}|D,\vec{\theta},\vec{m})=\prod_{i=1}^n\prod_{j=1}^{v_i}\Pr(Z_{ij}=z_{ij}|D,\vec{\theta})=\prod_{i=1}^n\prod_{j=1}^{v_i}\left\{\prod_{\ell\in\mathcal{R}_{ij}}(1-\theta_\ell)\right\}^{z_{ij}-1}\left\{1-\prod_{\ell\in\mathcal{R}_{ij}}(1-\theta_\ell)\right\}. 
\] 

Here  $Z_{ij}=\min\{W_{i\ell}\}_{\ell\in\mathcal{R}_{ij}}$
represents the minimum of the original geometric latent variables for
the entities available at stage $j$ in the $i$th observation. Equivalently, $Z_{ij}$ is the number of Bernoulli trials from the $j-1$th
event to the $j$th event in the $i$th observation, where an event is a
success for \textit{any} of the entities and where the $0$th event
occurs at trial 0.

It follows that the complete data likelihood is
\begin{align*}
p(D,\vec{Z}|\vec{\theta},\vec{m})&=p(D|\vec{\theta},\vec{m})p(\vec{Z}|D,\vec{\theta},\vec{m})\\
&=\prod_{k=1}^K \theta_k^{w_k}
  (1-\theta_k)^{d_k}\prod_{i=1}^n\prod_{j=1}^{v_i}\prod_{\ell\in\mathcal{R}_{ij}}(1-\theta_\ell)^{z_{ij}-1}\\
&=\prod_{k=1}^K \theta_k^{w_k}
  (1-\theta_k)^{ \zeta_k-w_k},
\end{align*}
where $\zeta_k=\sum_{i=1}^n\sum_{j=1}^{v_i} \delta_{ijk}z_{ij}$.

\subsubsection{Prior distribution}
\label{sec:prior}

The form of the complete data likelihood suggests the following prior
specification. Prior uncertainty regarding the entity parameters $\theta_k$
for $k\in\mathcal{K}$ is expressed  
by independent beta distributions, $\theta_k \sim
\text{Beta}(a_k,b_k)$, where $a_k,b_k>0$, such that 
\[
p(\theta_k)=\frac{\Gamma(a_k+b_k)}{\Gamma(a_k)\Gamma(b_k)}\theta_k^{a_k-1}(1-\theta_k)^{b_k-1},
\]
and $\Gamma(x)=\int_0^\infty y^{x-1}e^{-y}\text{d}y$ denotes the gamma function. 

\subsubsection{Posterior distribution and full conditional distributions}
\label{sec:posterior}

Based on the prior distribution and the complete data likelihood, the full joint probability (density) of parameters, data and
latent variables is 
\begin{align*}
p(D,\vec{Z},\vec{\theta}|\vec{m})&=p(D,\vec{Z}|\vec{\theta},\vec{m})p(\vec{\theta})\\
                         &=\prod_{k=1}^K\frac{\Gamma(a_k+b_k)}{\Gamma(a_k)\Gamma(b_k)}\theta_k^{a_k+w_k-1}(1-\theta_k)^{b_k+\zeta_k-w_k-1},
\end{align*}
and therefore the full conditional distribution for $\theta_k$ is 
\[
p(\theta_k|D,\vec{Z},\vec{\theta}_{-k},\vec{m})\propto
\theta_k^{a_k+w_k-1}(1-\theta_k)^{b_k+\zeta_k-w_k-1}, \qquad k=1,\ldots,K,
\]
which implies that
\[
\theta_k|D,\vec{Z},\vec{\theta}_{-k},\vec{m} \sim \text{Beta}(a_k+w_k,b_k+\zeta_k-w_k), \qquad k=1,\ldots,K, 
\]
where $\vec{\theta}_{-k}$ denotes the vector $\vec{\theta}$ with the
$k$th element removed.

\subsubsection{Gibbs sampler}
\label{sec:gibbs}

The following simple Gibbs sampler can be used to sample
$\vec{\theta},\vec{Z}$ with joint posterior density
$p(\vec{\theta},\vec{Z}|D,\vec{m})$. Starting with some initial value
$\vec{\theta}^{(0)}$, for $t=1,2,\ldots$ 
\begin{itemize}
\item sample $Z_{ij}^{(t)}|D,\vec{\theta}^{(t-1)},\vec{m} \sim
\text{Geom}\left(1-\prod_{\ell\in\mathcal{R}_{ij}}\left\{1-\theta_\ell^{(t-1)}\right\}\right)$
for $i=1,2,\ldots,n$ and $j=1,\ldots,v_i$,
\item sample $\theta_k^{(t)}
  |D,\vec{Z}^{(t)},\vec{m} \sim
  \text{Beta}\left(a_k+w_k,b_k+\zeta_k^{(t)}-w_k\right)$ for $k=1,\ldots,K$,
\end{itemize}
where $\zeta_k^{(t)}=\sum_{i=1}^n\sum_{j=1}^{v_i} \delta_{ijk}z_{ij}^{(t)}$.

\subsubsection{EM algorithm}
\label{sec:em}

A simple EM algorithm for finding the posterior mode (maximum
\textit{a posteriori} (MAP) estimator) can also be derived.  At
iteration $t$ of the EM algorithm set 
\[
\vec{\theta}^{(t)}=\underset{\vec{\theta}}{\text{argmax}}\,Q(\vec{\theta},\vec{\theta}^{(t-1)}),
\]
where
\begin{align*}
Q(\vec{\theta},\vec{\theta}^*)&=E_{\vec{Z}|D,\vec{\theta}^*}[\log
\pi(D,\vec{Z},\vec{\theta}|\vec{m})]\\
&=\sum_{k=1}^K
  w_k\log(\theta_k)+\left(\sum_{i=1}^n\sum_{j=1}^{v_i}\delta_{ijk}E[Z_{ij}|D,\vec{\theta}^*,\vec{Z}\setminus{\{Z_{ij}\}}]-w_k\right)\log(1-\theta_k)\\
  &\qquad+\sum_{k=1}^K(a_k-1)\log(\theta_k)+(b_k-1)\log(1-\theta_k)+\log(\Gamma(a_k+b_k))-\log(\Gamma(a_k))-\log(\Gamma(b_k))\\
&=C+\sum_{k=1}^K
  (a_k+w_k-1)\log(\theta_k)+\left(b_k+\sum_{i=1}^n\sum_{j=1}^{v_i}\frac{\delta_{ijk}}{1-\prod_{\ell\in\mathcal{R}_{ij}}(1-\theta_\ell^*)}-w_k-1\right)\log(1-\theta_k),
\end{align*}
where $C$ does not depend on $\vec{\theta}, \vec{\theta}^*$. 

Maximising $Q(\vec{\theta},\vec{\theta}^*)$ implies that at iteration $t$ of the EM algorithm  set
\begin{equation}
\theta_k^{(t)}
=\frac{a_k+w_k-1}{a_k+b_k+\sum_{i=1}^n\sum_{j=1}^{v_i}\frac{\delta_{ijk}}{1-\prod_{\ell\in\mathcal{R}_{ij}}\left\{1-\theta_\ell^{(t-1)}\right\}}-2},
\label{eq:em-map}
\end{equation}
for $k=1,\ldots,K$. Each iteration will lead to an increase in the log
posterior density and hence  
$\lim_{t \rightarrow\infty}\vec{\theta}^{(t)}=\hat{\vec{\theta}}$, the
unique MAP estimate. 

The MLE can be obtained by setting $a_k=1$, $b_k=1$
in~\eqref{eq:em-map};  the MLE is equivalent to the MAP estimator when
$\theta_k\overset{\text{indep.}}{\sim} U(0,1)$ for all $k$.

\subsection{Algorithms for paired comparisons}
\label{sec:gibbs-em-pairs}

For the special case of paired comparisons there are alternative
algorithms --- again inspired by those in \cite{CaronD12} --- which may be
more efficient in certain circumstances. 

\subsubsection{Likelihood based on paired comparisons}
\label{sec:pairs-like}

Consider the case of $n$ paired comparisons between $K$
  entities; therefore $m_i=n_i=2$ for all $i=1,\ldots,n$. Starting
  from~\eqref{eq:like-top}, the likelihood can be written as  
\begin{align*}
L(D|\vec{\theta}) &=\frac{\prod_{i\in\mathcal{K}} \theta_i^{w_i}(1-\theta_i)^{d_i}}{\prod_{i=1}^n
  \prod_{j=1}^{v_i}\left\{
    1-\prod_{\ell\in\mathcal{R}_{ij}}(1-\theta_{\ell})\right\}}\\
&=\frac{\prod_{i\in\mathcal{K}} \theta_i^{w_i}(1-\theta_i)^{c_i-w_i}}{\prod_{\{i,j\in \mathcal{K}; i<j\}}
  \left\{
    1-(1-\theta_i)(1-\theta_j)\right\}},
\end{align*}
where  $w_i$ can now be interpreted as the number
of wins plus ties for entity $i$ and $d_i$ can be interpreted as the
number of losses for entity $i$ (with $c_i$ still denoting the number
of comparisons involving entity $i$). Note that for the reverse GPL $w_i$
becomes the number of ties and losses for entity $i$ and $d_i$ becomes
the number of (outright) wins for entity $i$. 

\subsubsection{Alternative latent variables and complete data likelihood}
\label{sec:pairs-latent}

The usual latent variables for the GPL model when $n_i=2$ are, for $i=1,\ldots,n$
\[
Z_i | \vec{D},\vec{\theta} \sim \text{Geom}\left(1-\prod_{\ell\in\mathcal{R}_{i1}}(1-\theta_\ell)\right)
\]
that is,
\[
Z_i | \vec{D},\vec{\theta} \sim \text{Geom}(1-(1-\theta_{y_{i1}})(1-\theta_{y_{i2}})).
\]
Here
$Z_i=\text{min}(W_{iy_{i1}},W_{iy_{i2}})$. 

Now let
$n_{ij}$ denote the number of comparisons involving entities $i$ and
$j$. Also let
\[
\mathcal{Z}_{ij}=\sum_{\ell=1}^n Z_i\mathbb{I}(i\in\mathcal{R}_{\ell
  1} \cap j\in\mathcal{R}_{\ell 1})
\]
for $\{i,j\in \mathcal{K}; i<j\}$ denote the sum of the latent
geometric variables for comparisons between entities $i$ and $j$. Then for $\{i,j\in \mathcal{K}; i<j\}$,
\[
\mathcal{Z}_{ij}|\vec{D},\vec{\theta} \sim \text{NegBin}(n_{ij},1-(1-\theta_i)(1-\theta_j))
\]
independently for $n_{ij}>0$ and $\mathcal{Z}_{ij}=0$ if $n_{ij}=0$. 

The
full conditional probability
mass function of $\mathcal{Z}_{ij}$ for $\{i,j\in \mathcal{K};
(i<j)\cap n_{ij}>0\}$ is 
\[
\Pr(\mathcal{Z}_{ij}=z_{ij}|\vec{D},\vec{\theta})=\binom{z_{ij}-1}{n_{ij}-1}\{1-(1-\theta_i)(1-\theta_j)\}^{n_{ij}}\{(1-\theta_i)(1-\theta_j)\}^{z_{ij}-n_{ij}},
\]
for $z_{ij}=n_{ij},n_{ij}+1,\ldots$ and $n_{ij}=1,2,\ldots$. 

These negative binomial latent variables can be thought of as the sum
over all comparisons of the ``non-losing'' (for the standard smaller
is better GPL model, or ``non-winning'' for the reverse GPL model)
latent geometric random variables. This is also analogous to the use
of gamma latent variables in the Bayesian analysis of the
Bradley-Terry model proposed by \cite{CaronD12}. This formulation
involves at most $K(K-1)/2$ latent variables compared to $n$ in the
previous formulation of Section~\ref{sec:latent}.  The negative
binomial framework will be computationally and statistically more
efficient whenever $K(K-1)/2 < n$.  For a single round robin
tournament with $K$ competitors there are $n=K(K-1)/2$ comparisons and
so no efficiency gains should be expected. However, for
something like an English Premier League football season which is round robin
both home and away there are $n=K(K-1)$ comparisons and this negative
binomial framework would be expected to be more efficient. 

The full conditional for $\vec{\mathcal{Z}}$ (the collection
of $\mathcal{Z}_{ij}$ for $\{i,j\in \mathcal{K};
(i<j)\cap n_{ij}>0\}$) is 
\begin{align*}
p(\vec{\mathcal{Z}}&|D,\vec{\theta})\\ 
&= \prod_{\{i,j\in \mathcal{K};
(i<j)\cap
n_{ij}>0\}}\binom{z_{ij}-1}{n_{ij}-1}\{1-(1-\theta_i)(1-\theta_j)\}^{n_{ij}}\{(1-\theta_i)(1-\theta_j)\}^{z_{ij}-n_{ij}}\\
&=\left[\prod_{i\in\mathcal{K}}(1-\theta_i)^{\sum_{j>i|n_{ij}>0}z_{ij}-n_{ij}+\sum_{j<i|n_{ij}>0}z_{ji}-n_{ji}}  \right]\prod_{\{i,j\in \mathcal{K};
(i<j)\cap
n_{ij}>0\}}\binom{z_{ij}-1}{n_{ij}-1}\{1-(1-\theta_i)(1-\theta_j)\}^{n_{ij}}.
\end{align*}

Therefore the complete data likelihood is
\begin{align*}
p(D,\vec{\mathcal{Z}}|\vec{\theta})&=p(D|\vec{\theta})p(\vec{\mathcal{Z}}|D,\vec{\theta})\\
&=\prod_{i\in\mathcal{K}}\theta_i
^{w_i}(1-\theta_i)^{-w_i+\sum_{j>i|n_{ij}>0} z_{ij} + \sum_{j<i|n_{ij}>0} z_{ji}}\prod_{\{i,j \in \mathcal{K} ; (i < j) \cap (n_{ij}>0)\}} \binom{z_{ij}-1}{n_{ij}-1},
\end{align*}
as the number of comparisons involving entity $i$, $c_i=\sum_{j>i}n_{ij}+\sum_{j<i}n_{ji}$.

\subsubsection{Alternative Gibbs sampler for paired comparisons}
\label{sec:pairs-gibbs}

Setting $\xi_i=\sum_{j>i|n_{ij}>0} z_{ij} + \sum_{j<i|n_{ij}>0}
z_{ji}$ for notational conciseness, the full joint probability
(density) of parameters, latent variables and data can be written explicitly as 
\begin{align*}
p(D,\vec{\mathcal{Z}},\vec{\theta})&=p(D,\vec{\mathcal{Z}}|\vec{\theta})p(\vec{\theta})\\
&=\prod_{i\in\mathcal{K}}\frac{\Gamma(a_i+b_i)}{\Gamma(a_i)\Gamma(b_i)}\theta_i^{a_i+w_i-1}(1-\theta_i)^{b_i+\xi_i-w_i-1}\prod_{\{i,j \in \mathcal{K}; (i < j) \cap (n_{ij}>0)\}} \binom{z_{ij}-1}{n_{ij}-1},
\end{align*}
whence the full conditional density for $\theta_i$ is clearly
\[
p(\theta_i|D,\vec{\mathcal{Z}},\vec{\theta}_{-i}) \propto  \theta_i^{a_i+w_i-1}(1-\theta_i)^{b_i+\xi_i-w_i-1},
\]
therefore
\[
\theta_i|D,\vec{\mathcal{Z}},\vec{\theta}_{-i} \sim \text{Beta}\left(a_i+w_i,b_i+\xi_i-w_i\right).
\]

The following simple Gibbs sampling algorithm can be used to sample
values of $\vec{\theta},\vec{\mathcal{Z}}$ with joint
probability 
density function $p(\vec{\theta},\vec{\mathcal{Z}}|D)$. Starting with some initial value
$\vec{\theta}^{(0)}$, for $t=1,2,\ldots$ 
\begin{itemize}
\item sample $\mathcal{Z}_{ij}^{(t)}|D,\vec{\theta}^{(t-1)} \sim
\text{NegBin}\left(n_{ij},1-(1-\theta_{i}^{(t-1)})(1-\theta_{j}^{(t-1)})\right)$
for $\{i,j \in \mathcal{K} ; (i < j) \cap (n_{ij}>0)\}$ and
compute $\xi_i^{(t)}=\sum_{j>i|n_{ij}>0} z_{ij}^{(t)} + \sum_{j<i|n_{ij}>0}
z_{ji}^{(t)}$ for $i\in\mathcal{K}$, 
\item sample $\theta_i^{(t)}
  |D,\vec{\mathcal{Z}}^{(t)},\vec{\theta}^{(t-1)}_{-i} \sim
  \text{Beta}\left(a_i+w_i,b_i+\xi_i^{(t)}-w_i\right)$ for $i\in\mathcal{K}$.
\end{itemize}

\subsubsection{Alternative EM algorithm for paired comparisons}
\label{sec:pairs-em}

Following the derivation of the EM algorithm in Section~\ref{sec:em},
an alternative EM algorithm for finding the MAP estimate based on the negative binomial latent
variables of Section~\ref{sec:pairs-latent} has updates
\begin{align}
\theta_i^{(t)}
&=\frac{a_i+w_i-1}{a_i+b_i+\sum_{j>i|n_{ij}>0}\frac{n_{ij}}{1-(1-\theta_i^{(t-1)})(1-\theta_j^{(t-1)})} + \sum_{j<i|n_{ij}>0}
  \frac{n_{ji}}{1-(1-\theta_j^{(t-1)})(1-\theta_i^{(t-1)})}-2},
\label{eq:pairs-em}
\end{align}
for $i\in\mathcal{K}$. It follows that the MLE can be obtained when setting $a_i=b_i=1$ into
the above update equation~\eqref{eq:pairs-em}.

\subsection{Posterior and posterior predictive summaries}
\label{sec:post-sum}

Suppose $N$ samples from the posterior distribution
$p(\vec{\theta}|D)$ are available then the usual estimates of posterior
quantities of interest can be obtained. The main quantities used in
the examples of Section~\ref{sec:examples} are described below. 

\subsubsection{Marginal posterior means and credible intervals}
\label{sec:postmean}
 
Posterior means
$\bar{\vec{\theta}}=(\bar{\theta}_1,\ldots,\bar{\theta}_K)'$ are
estimated from the posterior samples $\{\vec{\theta}^{(t)}\}_{t=1}^N$
in the usual way, $\bar{\theta}_k = \frac{1}{N}\sum_{t=1}^N
\theta_k^{(t)}$, for $k\in\mathcal{K}$. Equal-tailed 95\% credible
intervals are also estimated from the posterior samples in the usual
way, with the upper (lower) limit being estimated by the upper (lower)
2.5\% empirical quantile.  

\subsubsection{Determining the best or most preferred entity}
\label{sec:bestentity}

If a purpose of the analysis is to determine the ``best'' entity then
there are several ways this can be approached. Recall that for the
standard \textit{smaller is better} GPL model the larger parameter
values lead to higher ranked entities. Therefore, the entity with the
largest posterior mean is an appropriate candidate for the best
entity. The entity with the largest posterior mode (MAP),
$\hat{\vec{\theta}}=(\hat{\theta}_1,\ldots,\hat{\theta}_K)'$, is also an
appropriate candidate. For the reverse GPL model, the smallest values
of the parameter values are best. 

There are, however, benefits to using the posterior mean when using
the standard GPL model, as explained below. The entity with the
largest posterior mean is also the one that maximises the
posterior predictive probability of being ranked first on its own in a
complete rank ordering under the GPL model. This can be seen by
examining the form of the probability of entity $k$ being ranked first
on its own in a complete rank ordering under the GPL model which was given
earlier in~\eqref{eq:entityk-first} and is reproduced below,
\begin{equation*}
p_k = \frac{\theta_k
  \prod_{i\in\mathcal{K}\setminus{\{k\}}}(1-\theta_i)}{1-\prod_{j\in\mathcal{K}}(1-\theta_j)}
\end{equation*}
where $\mathcal{K}\setminus{\{k\}}$ denotes the set $\mathcal{K}$ with
the $k$th element removed.  
Therefore, an estimate of the posterior
predictive probability $p_k|D$ is $\hat{p}_{k}=\frac{1}{N}\sum_{t=1}^N
p_k^{(t)}$, with $p_k^{(t)}$ obtained by
calculating~\eqref{eq:entityk-first} based on the sampled values
$\vec{\theta}^{(t)}$ at iteration $t$.

Under the reverse GPL model an explicit expression for the probability
of entity $k$ being ranked in $K$th position in a complete rank
ordering (and hence ``best'') is not available in closed form. This
probability can be estimated by simulating synthetic complete rank
orderings from the posterior predictive distribution and calculating
the proportion of the synthetic posterior predictive samples in which
entity $k$ is ranked in $K$th position on its own. Sampling from the
posterior predictive distribution like this allows a plethora of other
quantities of interest to be estimated.

\subsubsection{Determining a total order of the entities: rank aggregation}
\label{sec:total-order}

If a total order (an ordering without ties) of a subset of entities
$\mathcal{K}^\star \subseteq\mathcal{K}$ is desired then the ordering
of the posterior means of the entities in $\mathcal{K}^\star$ from
largest to smallest,
$\hat{\vec{y}}=\text{order}_{\downarrow}\{\bar{\theta}_k\}_{k\in\mathcal{K}^\star}$,
provides a simple estimate under the GPL model. For the reverse GPL
model the
ordering is taken from smallest to largest.  From a Bayesian
perspective, the ordering of the entities which has maximum posterior
predictive probability is the optimal total order, since no other
total order has higher posterior predictive probability. Given the
number of possible total orders is $|\mathcal{K}^\star|!$ an
exhaustive search is only possible for small $|\mathcal{K}^\star|$.
For moderate $|\mathcal{K}^\star|$ algorithms such as the cyclic
coordinate ascent algorithm described in \cite{johnson2020revealing}
are effective; see \cite{johnson2020revealing} and
\cite{johnson2022bayesian} for more discussion of posterior predictive
summaries for ranked data without ties. Ordering the posterior means
of the parameters can be considered as a first order approximation to
the modal order under the posterior predictive distribution, is far
simpler to compute and scales up easily to large values of $K$; see
Section~2.2 in the Supplementary Material of
\cite{johnson2020revealing} for further details. Ordering the
posterior means of the parameters will therefore be the primary method
for estimating an optimal total order in the results presented in Section~\ref{sec:examples}.

\section{Examples}
\label{sec:examples}

\subsection{Introduction}
\label{sec:ex-intro}

In this section three real datasets are used to illustrate several
features of the GPL model and the inference algorithms described in
Section~\ref{sec:inference}. The examples have been chosen to cover a
range of data types --- paired comparisons, complete and top-$m$ rank
orderings --- involving either preferences or competition. The
datasets cover a range of values of $n$, the number of observations,
and $K$, the number of entities; see Table~\ref{tab:summary-real} for
a summary.

\begin{table}[tbh]
\centering
\begin{tabular}{lrrlrrr}
\hline
Data set & $n$ & $K$ & Type & CPU(s) & Min(ESS) & Min(ESS)/CPU \\ 
\hline
Puddings & 745 & 6 & Paired & 3.6 & 2839 & 789 \\
NASA & 10 & 32 & Complete & 13 & 6196 & 477 \\
Golf & 46 & 631 & Top-$m$ on subsets & 1158 & 7456 & 6 \\
\end{tabular}
\caption{Summary of datasets used in the examples together with
  example timings and effective sample sizes based on 10000 Gibbs sampler
  iterations for the GPL model.}
\label{tab:summary-real}
\end{table}

In order to avoid repetition of similar details, the general choices
made in obtaining the results are detailed below. Clearly these are not
generic recommendations and choices should be tailored to specific
examples. 

\subsubsection{Computational details}
\label{sec:ex-comp}

The Gibbs samplers and EM algorithms of Section~\ref{sec:inference}
were coded in \textsf{R}~\citep{R}; data and example code for reproducing
the presented results is available from \url{www.github.com/d-a-henderson/GPL}.
The code was run on a fairly standard desktop computer with an
Intel(R) Core(TM) i7-9700 CPU @ 3.00GHz, 18GB RAM.

\subsubsection{Prior distributions} 
\label{sec:ex-prior}

 For illustrative purposes the prior distribution for the parameters
 of the GPL model 
in each example was $\theta_k\overset{\text{indep.}}{\sim}
\text{Beta}(1,1)$, that is $a_k=b_k=1$ for $k=1,\ldots,K$. This
exchangeable specification with uniformly distributed parameters
represents a convenient baseline. Clearly
more specific prior information could be encoded if available.

\subsubsection{Gibbs sampling}
\label{sec:ex-gibbs}

In all cases, four chains were generated via the Gibbs sampler and the
potential scale reduction factor of \cite{GelmanR92} was computed
using the \texttt{gelman.diag} function in the \textsf{coda} package~\citep{coda} and gave values for
all parameters of 1.00 to two decimal places indicating no evidence of
non-convergence. Convergence and mixing were also assessed via the
usual visual and numerical checks. In all examples reported here very
little burn-in is required and in all cases discarding the first 10
iterations was sufficient. After the burn-in period each chain
was run for $N=10000$ iterations. Mixing was generally very good with
autocorrelations becoming indistinguishable from zero by lag 10 in the
worst case.  The results from a single chain of 10000 iterations are
presented.

\subsubsection{EM algorithm}
\label{sec:ex-em}

The EM algorithm was initialised at a random draw from the prior
distribution. It is not necessary to run the EM algorithm from
multiple starting points as the likelihood and therefore the posterior
distribution (under the prior in Section~\ref{sec:prior}) has a unique
mode. The algorithm is stopped when the mean squared difference
between parameter values at successive iterations is less than
$10^{-16}$.

\subsection{Paired comparisons: puddings}
\label{sec:puddings}

\subsubsection{Background and data}
\label{sec:mariner-background}

\cite{Davidson70} introduced and analysed a dataset of $n=745$ paired comparisons
involving $K=6$ brands of chocolate milk pudding. These data were also analysed recently by
\cite{turner2020modelling}. 

\subsubsection{Bayesian analysis via the GPL and reverse GPL}
\label{sec:puddings-GPL}

With $n=745$ being much greater than $K(K-1)/2=15$ the negative
binomial-based Gibbs sampler of Section~\ref{sec:pairs-gibbs} is
preferable to the geometric-based Gibbs sampler of
Section~\ref{sec:gibbs}.  For example, a single chain of 10010
iterations takes approximately 34 seconds for the Gibbs sampler of
Section~\ref{sec:gibbs} but only approximately 3.6 seconds for the
Gibbs sampler of Section~\ref{sec:pairs-gibbs} and mixing of the two
chains is similar for these data.  Consequently, the
Gibbs sampler of Section~\ref{sec:pairs-gibbs} was used to sample from
the posterior distribution to obtain the results presented here.
Table~\ref{tab:summary-real} shows, amongst other things, that the
minimum effective sample size (ESS, using the definition in the
\texttt{coda} package~\citep{coda} in \textsf{R}) was 2839.  The mixing of the chain
is illustrated in Figure~\ref{fig:pud-mix} which
shows a trace plot and an autocorrelation function plot for the
parameter corresponding to this minimum ESS and indicates that the
chain mixes well.
\begin{figure}[h]
\centering
\includegraphics[width=0.95\linewidth]{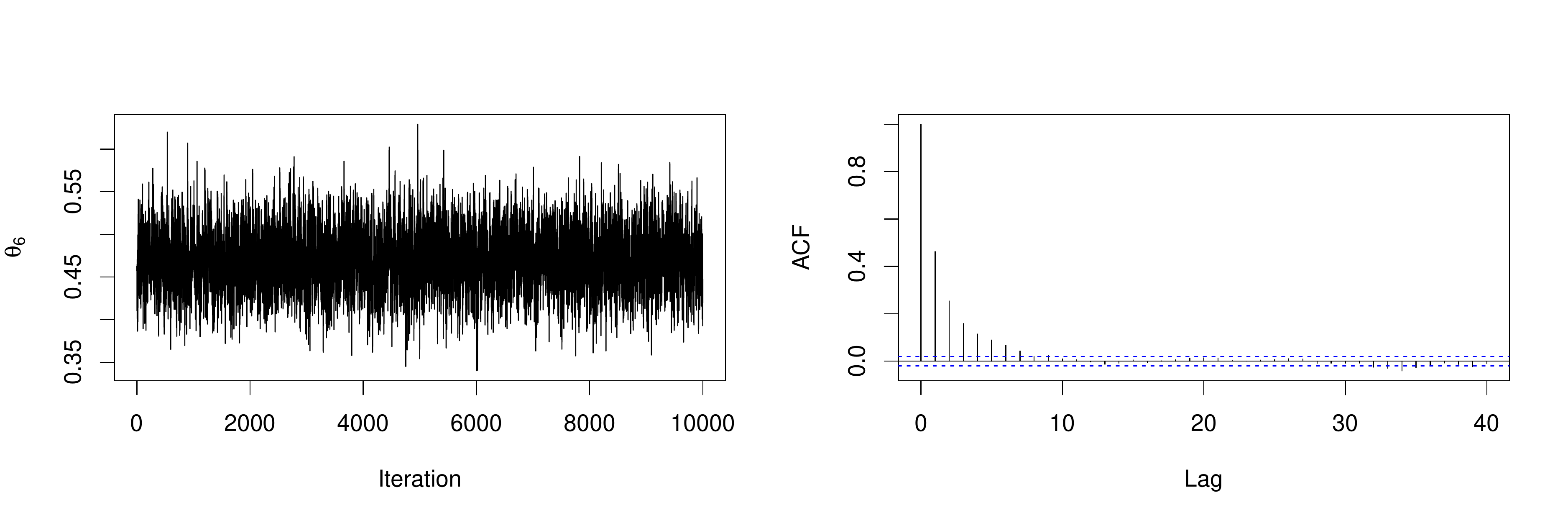}
\caption{Puddings data: Trace plot (left) and autocorrelation function plot (right)
  for 10000 posterior samples of $\theta_6$ obtained via the negative
  binomial-based Gibbs sampler of Section~\ref{sec:pairs-gibbs}. The
  samples have an ESS of 2839, the lowest of the $K=6$ parameters.}
\label{fig:pud-mix}
\end{figure}

Figure~\ref{fig:pud-postMAP} displays summaries of the marginal
posterior distributions of the parameters of the GPL model (left) and
the reverse GPL model (right), in
particular posterior means, modes and credible intervals. 
\begin{figure}[h]
\centering
\includegraphics[width=0.48\linewidth]{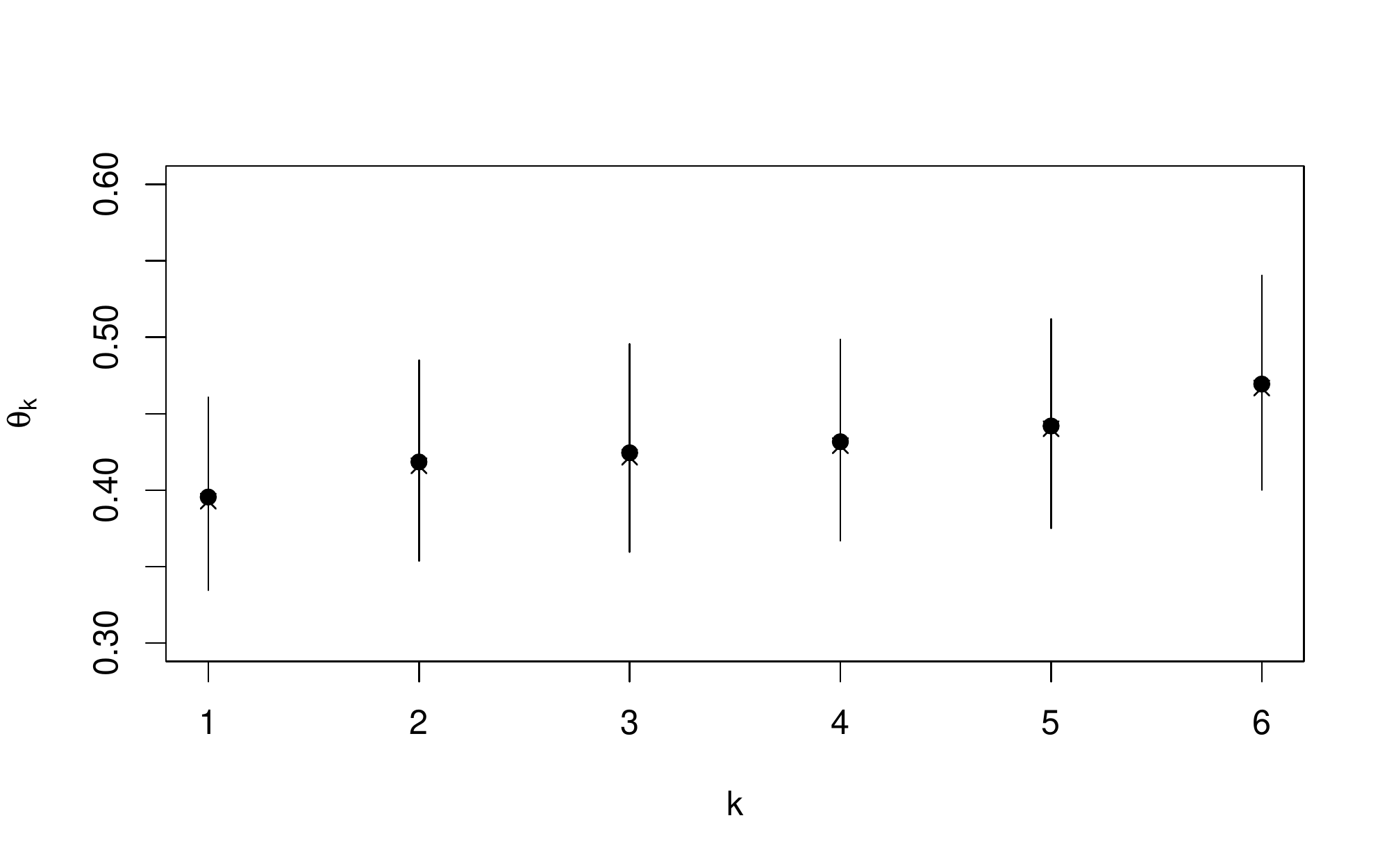}
\includegraphics[width=0.48\linewidth]{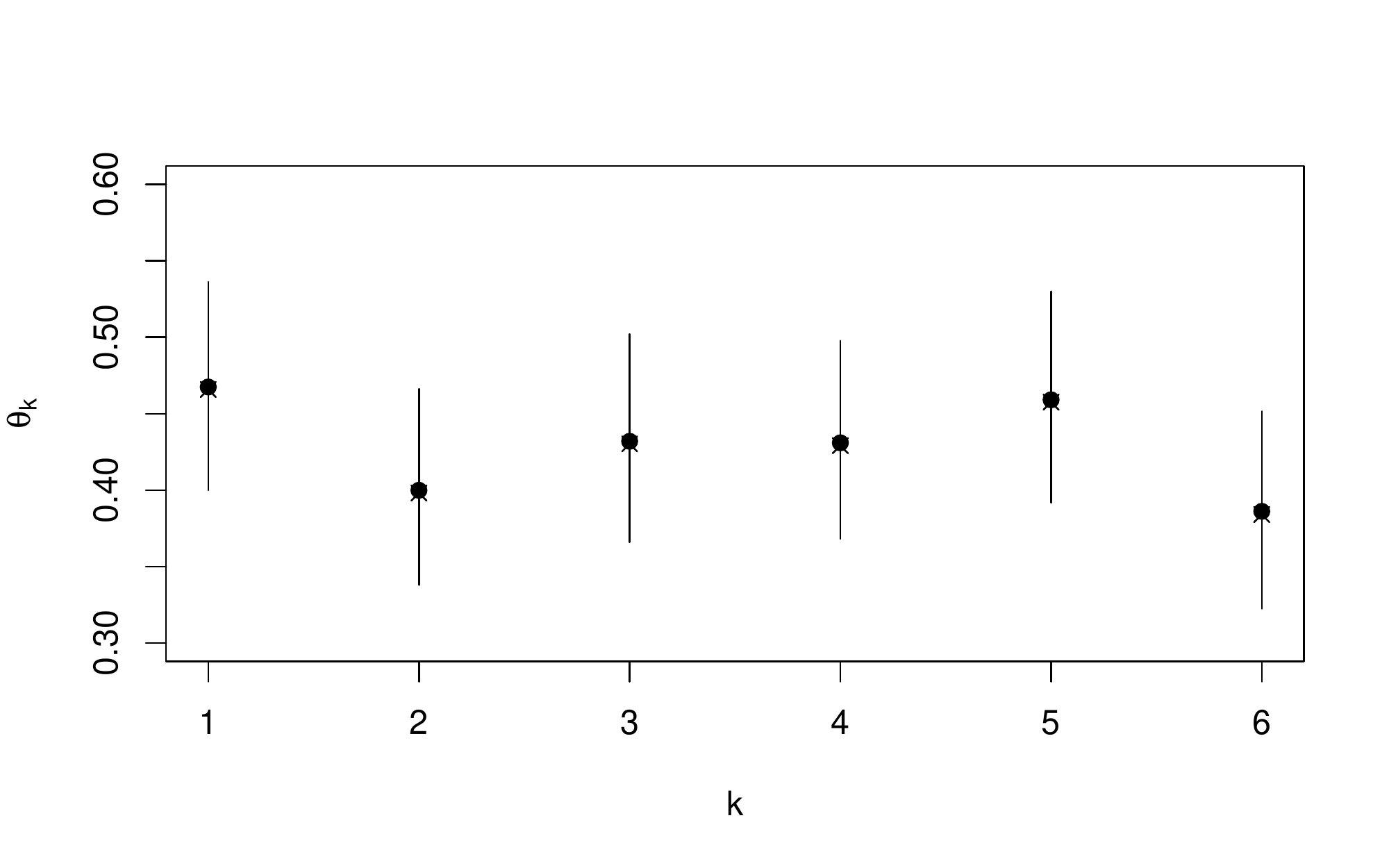}
\caption{Marginal posterior means (dots), MAP estimators/MLEs
  (crosses) and 95\% equal-tailed credible intervals (lines) for the
  parameters of the GPL model (left) and the reverse GPL model (right)
  applied to the puddings data.}
\label{fig:pud-postMAP}
\end{figure}
Posterior means and MAP estimates are also reported in Table
\ref{tab:pud-postpred}. The MAP estimates, which correspond to the MLEs
with the current choice of prior, were obtained via the EM algorithm of
Section~\ref{sec:pairs-em}; for the results reported here the
algorithm took 47 iterations to satisfy the stopping criterion. The
MAP estimates and the posterior means are very similar in this example.
There is also clearly considerable overlap in the credible intervals.
\begin{table}[h]
\centering
\begin{tabular}{lll}
\hline
GPL model & Estimator & Estimate \\ 
\hline
Standard  & Posterior mean & $\bar{\vec{\theta}}=(0.396, 0.418, 0.424, 0.432, 0.442, 0.469)'$ \\
 & MAP & $\hat{\vec{\theta}}=(0.393, 0.416, 0.422, 0.429, 0.440,
0.467)'$ \\
 & Total order & $\hat{\vec{y}}=(6, 5, 4, 3, 2, 1)'$\\
 & LOOIC & 1632.2 (7.9) \\
\hline
Reverse  & Posterior mean & $\bar{\vec{\theta}}=(0.467, 0.400, 0.432, 0.431, 0.459, 0.386)'$ \\
 & MAP & $\hat{\vec{\theta}}=(0.466, 0.398, 0.430, 0.429, 0.458, 0.384)'$ \\
 & Total order & $\hat{\vec{y}}=(6, 2, 4, 3, 5, 1)'$\\
 & LOOIC & 1629.9 (8.4) \\
\end{tabular}
\caption{Posterior and posterior predictive summaries for GPL models
  applied to the puddings data. The value in parentheses after the
  LOOIC estimate is its standard error.}
\label{tab:pud-postpred}
\end{table}

Table~\ref{tab:pud-postpred} also contains an estimate of the optimal total order
$\hat{\vec{y}}$ under each model, that is the aggregate rank ordering
without ties. These estimates were obtained by ordering the posterior
means of the parameters from largest to smallest for the GPL model,
and from smallest to largest for the reverse GPL model; see
Section~\ref{sec:total-order}. In fact, in each case these
\textit{are} the optimal total orders in terms of the posterior
predictive distribution, as they maximise the posterior predictive
probability under the respective model over all $6!$ possible total
orders. Interestingly, the total orders are different under the two
models, with entities 5 and 2 swapping in 2nd and 5th place.

There is, however, agreement under the two models as to the ``best'' brand of chocolate
pudding: brand 6.  Under the GPL model, this is the
entity with the largest posterior mean and therefore the entity which
maximises the posterior predictive probability of being ranked first
on its own in a complete rank ordering; see
Section~\ref{sec:bestentity}. Under the reverse GPL model, the
candidate for the optimal entity is that which has the smallest
posterior mean, in this case brand 6.

Table~\ref{tab:pud-postpred} also contains estimates of the
leave-one-out information criterion (LOOIC) which is advocated as a
trade-off between predictive fit and complexity of a model in \cite{vehtari2017practical} and is useful for model
comparison and selection.  Smaller values of LOOIC are better.
The values were computed using the \texttt{loo} package~\citep{loo} in \textsf{R}.  Here the reverse
GPL model has the marginally smaller LOOIC value although the two values are
statistically indistinguishable given the sampling variability (as
measured through the standard errors). 

Interestingly, a Bayesian analysis of the puddings data using the
Davidson model provides similar inferences to those under the reverse
GPL model; see Appendix~\ref{sec:puddings-davidson}.

\subsection{Complete rank orderings: NASA}
\label{sec:mariner}

\subsubsection{Background and data}
\label{sec:mariner-background}

\cite{DyerM76} describe a dataset of preferences of $n=10$ teams of
scientists tasked by NASA  to rank $K=32$ possible pairs of 
trajectories for the Mariner Jupiter/Saturn 1977 Project, which later
become the Voyager mission. The idea was to aggregate the preferences
from the science teams and select a single pair of trajectories for
the spacecraft.  

The NASA dataset is available from
\url{www.preflib.org/dataset/00003}. These are complete
rank orderings, so that $m_i=n_i=K$ for $i=1,\ldots,n$. Of the 10
teams, only two (teams 8 and 10) do not record any ties. Ties are
prevalent, and several ties of high order (the number of entities in
the tie) are observed, with an 8-way, a 10-way and even a 24-way tie.

This is a dataset that cannot currently be fitted using the
Davidson-Luce model.  When $K=32$, $|\Omega|=2^{32}-1=4294967295$, so
there are over 4 billion terms in the summation to get the normalising
constant. Even truncating the summation at ties of the maximum
observed order, as is done in the \texttt{PlackettLuce} \textsf{R}
package~\citep{turner2020modelling}, is not
enough since there is a 24-way tie. Indeed, \cite{turner2020modelling} suggests that the \texttt{PlackettLuce} package can only fit models with
ties up to order 4 and half of the teams in the NASA dataset report
ties with more than 4 entities. Fortunately, this dataset can be
handled easily and efficiently by the GPL model. 

\subsubsection{Bayesian analysis via the GPL and reverse GPL models}
\label{sec:mariner-bayes}

Table~\ref{tab:summary-real} shows, amongst other things, that 10000
iterations of the Gibbs sampler of Section~\ref{sec:gibbs} took about
13 seconds and the minimum ESS was 6196, indicating better mixing than
in the analysis of the puddings data.  The mixing of the chain is
illustrated in Figure~\ref{fig:nasa-mix} which shows a trace plot
and an autocorrelation function plot for the parameter corresponding
to this minimum ESS and indicates that the chain mixes well.
\begin{figure}[h]
\centering
\includegraphics[width=0.95\linewidth]{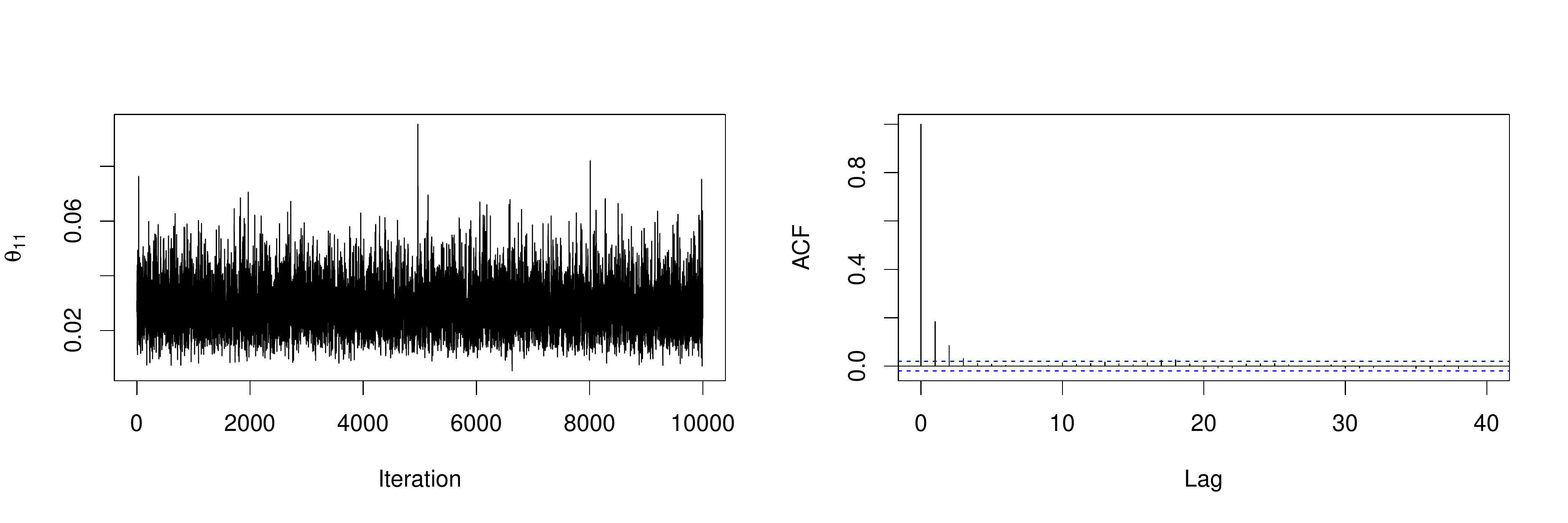}
\caption{NASA data: Trace plot (left) and autocorrelation function plot (right)
  for 10000 posterior samples of $\theta_{11}$ obtained via the Gibbs sampler of Section~\ref{sec:gibbs}. The
  samples have an ESS of 6196, the lowest of the $K=32$ parameters.}
\label{fig:nasa-mix}
\end{figure}

Figure~\ref{fig:nasa-postMAP} displays summaries of the marginal
posterior distributions of the parameters of the GPL model (left) and
the reverse GPL model (right), in
particular posterior means, modes and credible intervals. 
\begin{figure}[h]
\centering
\includegraphics[width=0.48\linewidth]{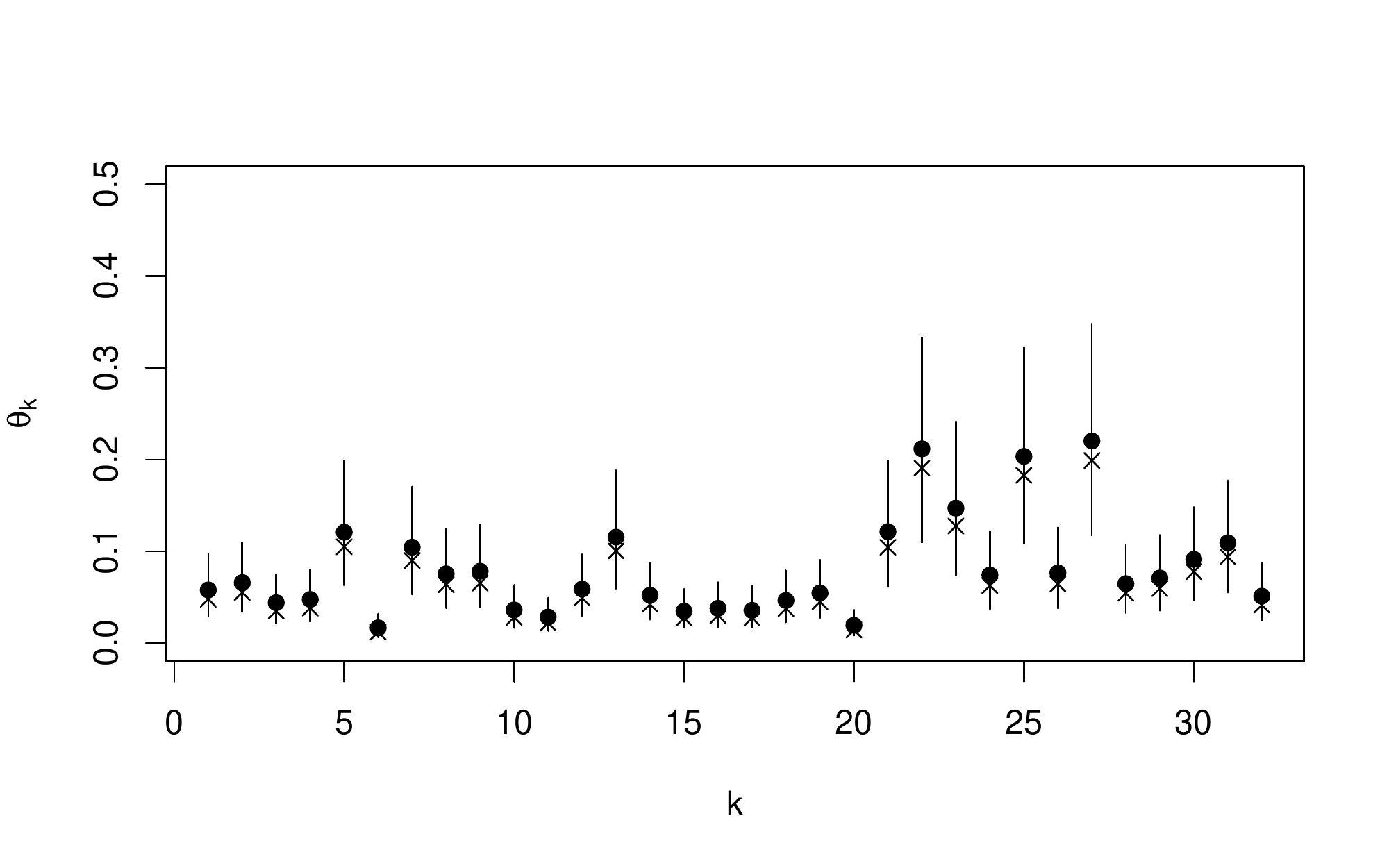}
\includegraphics[width=0.48\linewidth]{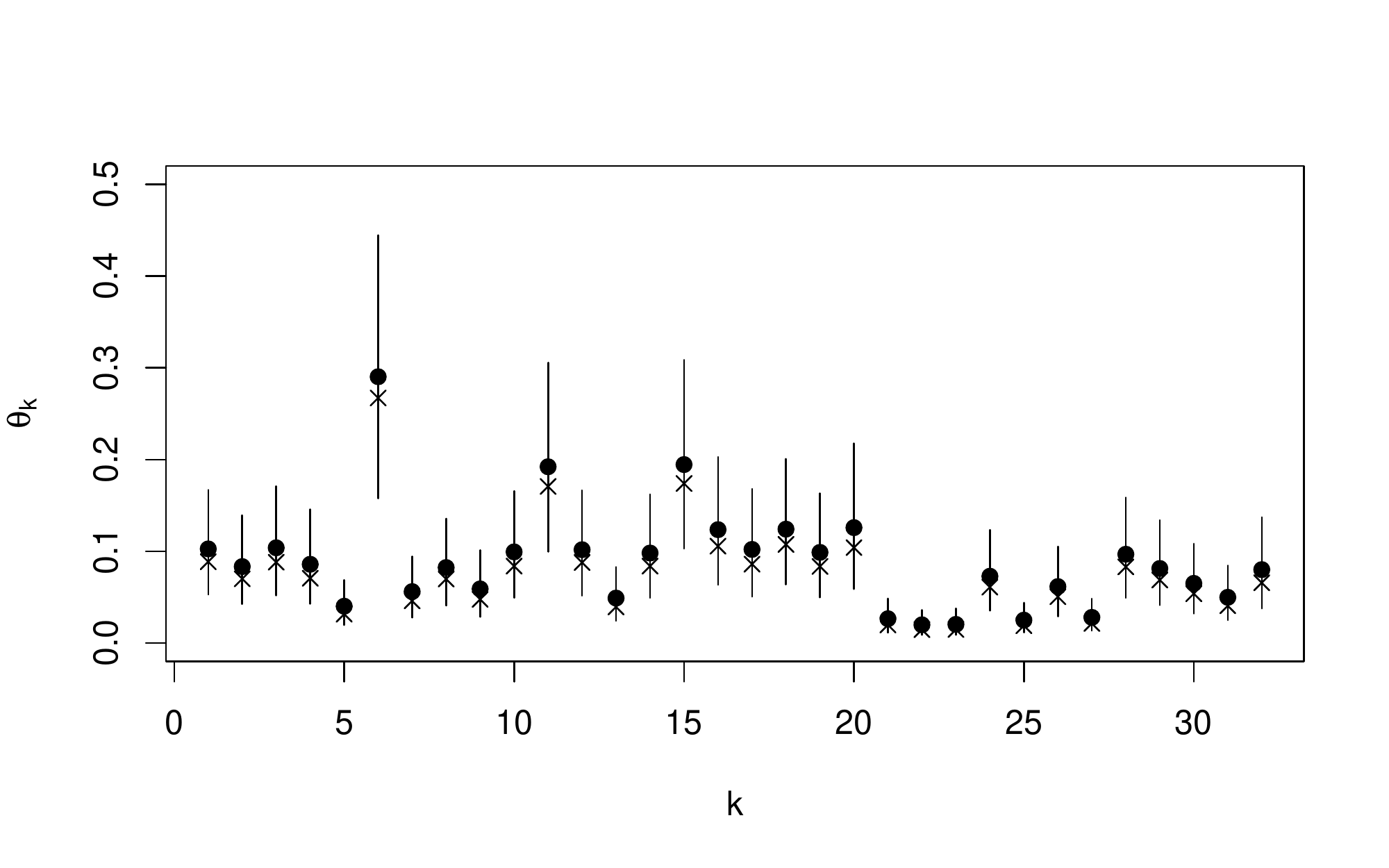}
\caption{Marginal posterior means (dots), MAP estimators/MLEs
  (crosses) and 95\% equal-tailed credible intervals (lines) for the
  parameters of the GPL model (left) and the reverse GPL model (right)
  applied to the NASA data.}
\label{fig:nasa-postMAP}
\end{figure}
The MAP estimates, which correspond to the MLEs
with the current choice of prior, were obtained via the EM algorithm of
Section~\ref{sec:em}; for the results reported here the
algorithm took 30 and 31 iterations, respectively for the GPL and
reverse GPL models,  to satisfy the stopping criterion. In contrast to
the puddings example, the credible intervals are not all overlapping;
the differences between the MAP and posterior mean estimates are also
more noticeable for the NASA data. The estimated parameter values
under the GPL model are generally
lower than those from the puddings example, perhaps reflecting the
prevalence of ties. 

Table~\ref{tab:nasa-postpred} contains an estimate of the optimal
total order $\hat{\vec{y}}$ under each model, that is, the aggregate
rank ordering without ties. These estimates were obtained by ordering
the posterior means of the parameters, as described in
Section~\ref{sec:total-order}. The total orders are quite different
under the two models, with, for example, entity 27 being first under
the GPL model but only 5th under the reverse GPL model.
\begin{table}[h]
\centering
\begin{tabular}{lll}
\hline
GPL model & Estimator & Estimate \\ 
\hline
Standard  & Total order & $\hat{\vec{y}}=(27, 22, 25, 23, 21,  5, 13, 31,  7,
30,  9, 26,  8, 24, 29,  2,$\\
          &             &  \qquad\quad $28, 12,  1, 19, 14, 32,  4, 18,  3, 16, 10, 17, 15, 11, 20, 6)'$\\
 & LOOIC & 1901.5 (29.6) \\
\hline
Reverse  & Total order & $\hat{\vec{y}}=(22, 23, 25, 21, 27,  5, 13, 31,  7,  9,
26, 30, 24, 32, 29,  8,$\\
  & &   \qquad\quad $2,  4, 28, 14, 19, 10, 12, 17,  1, 3, 16, 18, 20, 11, 15,  6)'$\\
 & LOOIC & 1828.8 (29.7) \\
\end{tabular}
\caption{Posterior predictive summaries for GPL models
  applied to the NASA data. The value in parentheses after the
  LOOIC estimate is its standard error.}
\label{tab:nasa-postpred}
\end{table}

This is a scenario where an aggregate ranking is not of primary
interest. The main interest lies in the choice of the optimal, most
preferred, trajectory pair. This is trajectory pair 27 under the GPL
model and trajectory pair 22 under the reverse GPL model. Note also
that trajectory pair 22 was placed second under the GPL model. The
LOOIC values in Table~\ref{tab:nasa-postpred} suggest that the reverse
GPL model is the better of the two models in terms of fit, even
accounting for the relatively large standard errors. If one of the two
models was to be selected for modelling these data then it would be
the reverse GPL model. It is reassuring to note that the optimal
trajectory pair under the reverse GPL model, trajectory pair 22, was the
trajectory pair selected by NASA (based on the process described in
\cite{DyerM76}) for the
Mariner Jupiter/Saturn 1977 mission.

\subsection{Top-$m$ on subsets of entities: Golf}
\label{sec:golf}

\subsubsection{Background and data}
\label{sec:golf-data}

The final dataset contains the results of the 47 golf tournaments from
the PGA Tour in 2021, including the four major tournaments: the Masters,
U.S.\ PGA, U.S.\ Open and The Open. The results were obtained from The
Official World Golf Rankings webpages (\url{www.owgr.com}) in
February~2022 but only the final finishing positions were recorded, rather than
the scores.  Golfers who were disqualified or withdrew from a
tournament were not included in the final rank ordering for that
tournament. Further details of the data and relevant information about
the game of golf is provided in Appendix~\ref{sec:golf-info}.

In what follows, the first $n=46$ tournaments are used to fit the GPL
model. Then predictions are made for the results of the final
tournament of the 2021 season, the Hero World Challenge. The dataset
includes $K=631$ golfers.  The data are in the form of top-$m$ rank
orderings as players who ``missed the cut'' in a particular tournament
are listed but not ranked. 

The maximum number of golfers in any one tournament is 156.
Figure~\ref{fig:uspga21-hist-tournaments} displays a histogram of the
number of tournaments contested by each of the 631 golfers. There is a
clear bimodal distribution with a group of around 200 golfers
appearing in  15 or more tournaments and a large number (248) of the
631 golfers appearing in only one tournament; the median number of
tournaments contested was 4.
\begin{figure}[h]
\centering
\includegraphics[width=0.6\linewidth]{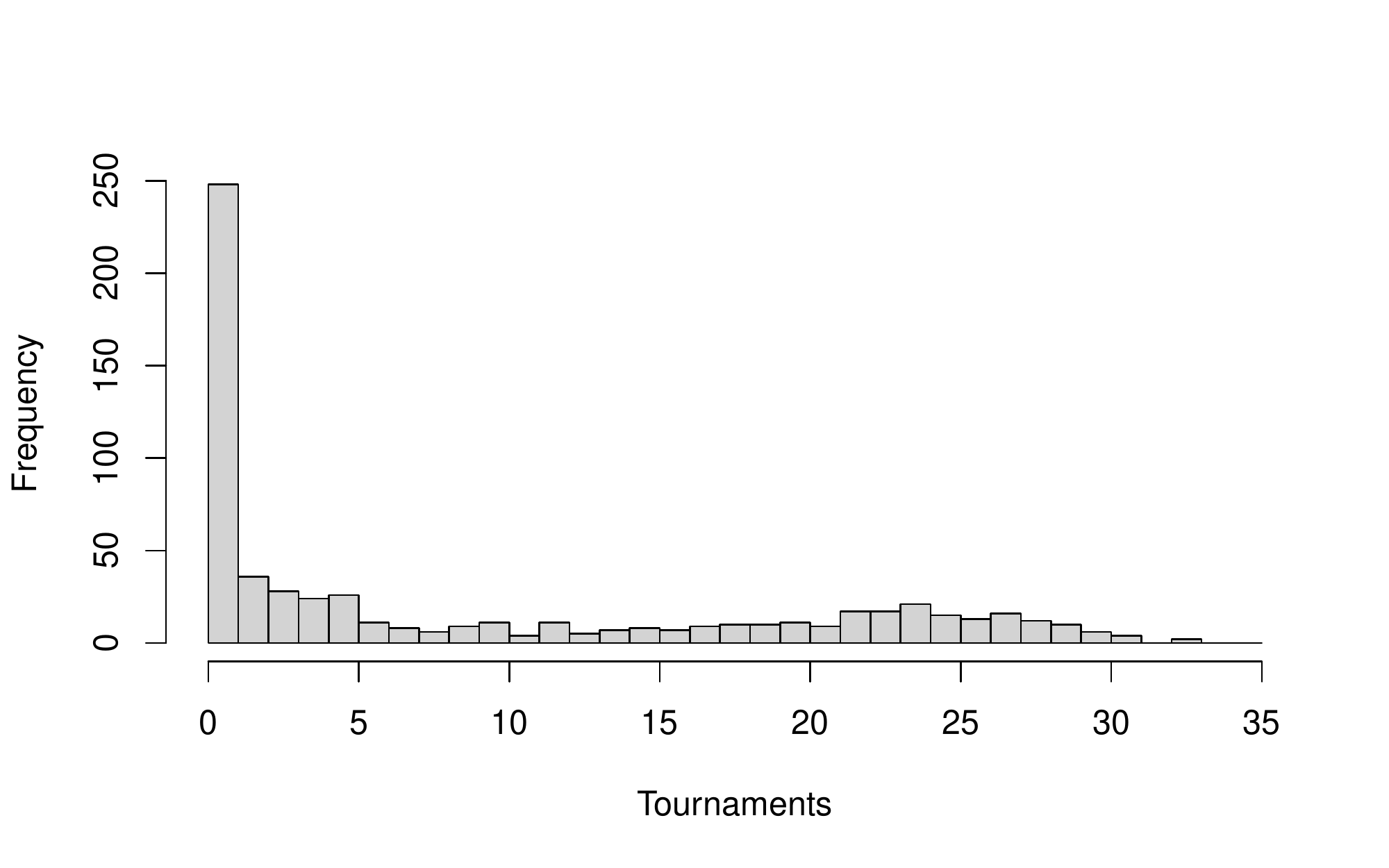}
\caption{Golf data: histogram of number of tournaments contested by
  $K=631$ golfers in 46 tournaments in 2021.}
\label{fig:uspga21-hist-tournaments}
\end{figure}

This is another example which is currently beyond the capabilities of
the Davidson-Luce model as ties of more than 4 golfers are relatively
common.  Fortunately, the GPL model is able to deal naturally with
both the top-$m$ rankings and the large number of ties. The generative
nature of the GPL model also proves to be very useful for predictive
inference. 

\subsubsection{Bayesian analysis via the GPL model}
\label{sec:golf-bayes-gpl}

The Gibbs sampler of Section~\ref{sec:gibbs} took about 20 minutes to
perform 10000 iterations and the minimum ESS was 7203, indicating
slightly better mixing than in the analysis of the NASA data; see
Table~\ref{tab:summary-real}.  As in the analyses of the other
datasets in this paper, posterior correlations between parameters are
typically low. The
good mixing of the chain is illustrated in Figure~\ref{fig:uspga21-mix} which
shows a trace plot and an autocorrelation function plot for the
parameter corresponding to the minimum ESS.
\begin{figure}[h]
\centering
\includegraphics[width=0.95\linewidth]{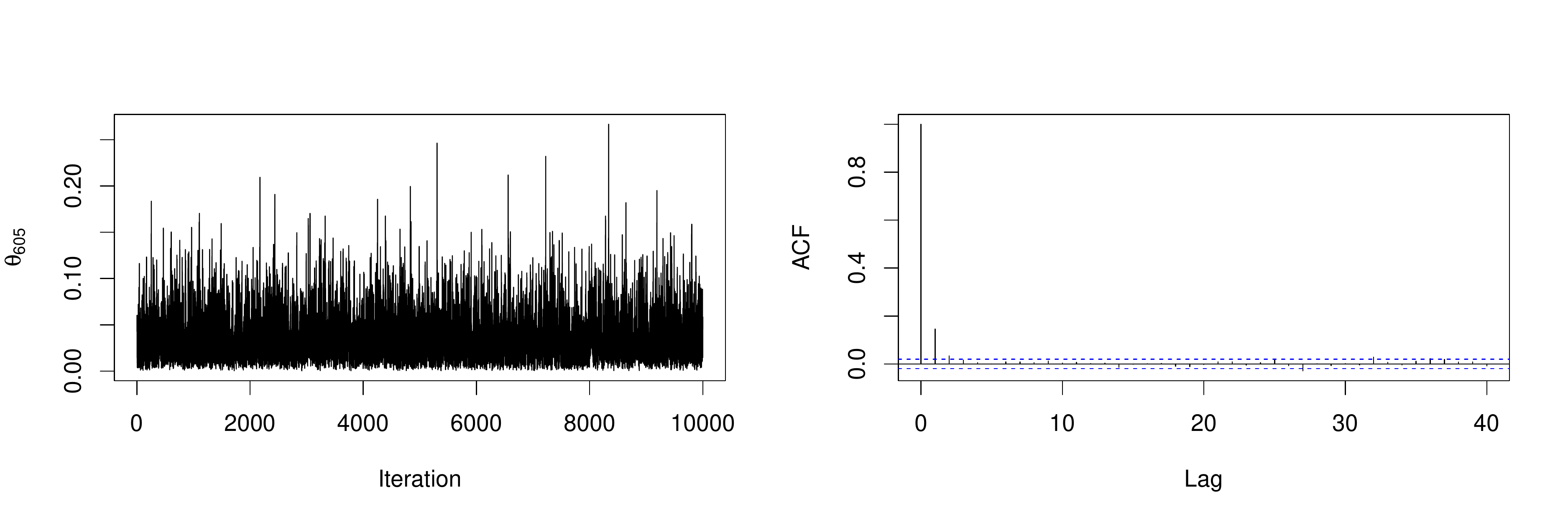}
\caption{Golf data: Trace plot (left) and autocorrelation function plot (right)
  for 10000 posterior samples of $\theta_{605}$ (corresponding to
  Ryuji Imada)  obtained via the Gibbs sampler of Section~\ref{sec:gibbs}. The
  samples have an ESS of 7456, the lowest of the $K=631$ parameters.}
\label{fig:uspga21-mix}
\end{figure}
This level of mixing is encouraging, given the large
number of parameters and latent variables, and provides evidence of
the statistical efficiency and scalability of the Gibbs sampler. 

The inferences presented below relate to a subset of 197 players who
played 15 or more tournaments, denoted $\mathcal{K}_{15+}$.
Table~\ref{tab:uspga21-top10} lists the top 10 players  ranked in
terms of the estimated posterior means $\bar{\theta}_k$ for
$k\in\mathcal{K}_{15+}$. 
\begin{table}[tbh]
\centering
\begin{tabular}{llrc}
\hline
Rank & Player & $\bar{\theta}_k$ & 95\% CI  \\
\hline
1 & Jordan Spieth      &  0.120 & (0.075, 0.175)\\
2 & Louis Oosthuizen   &  0.110 & (0.067, 0.163)\\
3 & Jon Rahm           &  0.109 & (0.064, 0.164)\\ 
4 & Collin Morikawa    &  0.106 & (0.065, 0.156)\\ 
5 & Daniel Berger      &  0.104 & (0.063, 0.152)\\ 
6 & Justin Thomas      &  0.102 & (0.062, 0.149)\\
7 & Bryson DeChambeau  &  0.100 & (0.061, 0.147)\\
8 & Xander Schauffele  &  0.098 & (0.060, 0.144)\\
9 & Viktor Hovland     &  0.097 & (0.061, 0.142)\\
10& Paul Casey         &  0.096 & (0.056, 0.144)\\ 
\end{tabular}
\caption{Top 10 golfers out of those who played 15 or more tournaments
   ranked in terms of estimated  posterior means
  $\bar{\theta}_k$. Equal-tailed 95\% credible intervals are also given.}
\label{tab:uspga21-top10}
\end{table}
Jordan Spieth is the top ranked golfer based on the GPL model but
there is considerable overlap in the credible intervals. A full total order
of all 197 golfers in $\mathcal{K}_{15+}$ is straightforward to obtain
but is not the main focus of this analysis. 

\subsubsection{Predictive inference for 2021 Hero World Challenge}
\label{sec:golf-bayes-pred}

The GPL model lends itself naturally to prediction of future outcomes
and this is illustrated here by using the data on the first $n=46$
tournaments to predict the outcome of the final tournament in 2021,
the Hero World Challenge which was held between the 2nd and 5th
December 2021 in Providence, Bahamas. This tournament, hosted by Tiger
Woods, involved 20 of the top players. Figure~\ref{fig:uspga21-hero-post}
displays summaries of the marginal posterior distributions for the 20
players taking part in the 2021 Hero World Challenge, ordered in terms
of their posterior means. 
\begin{figure}[tbh]
\centering
\includegraphics[width=0.8\linewidth]{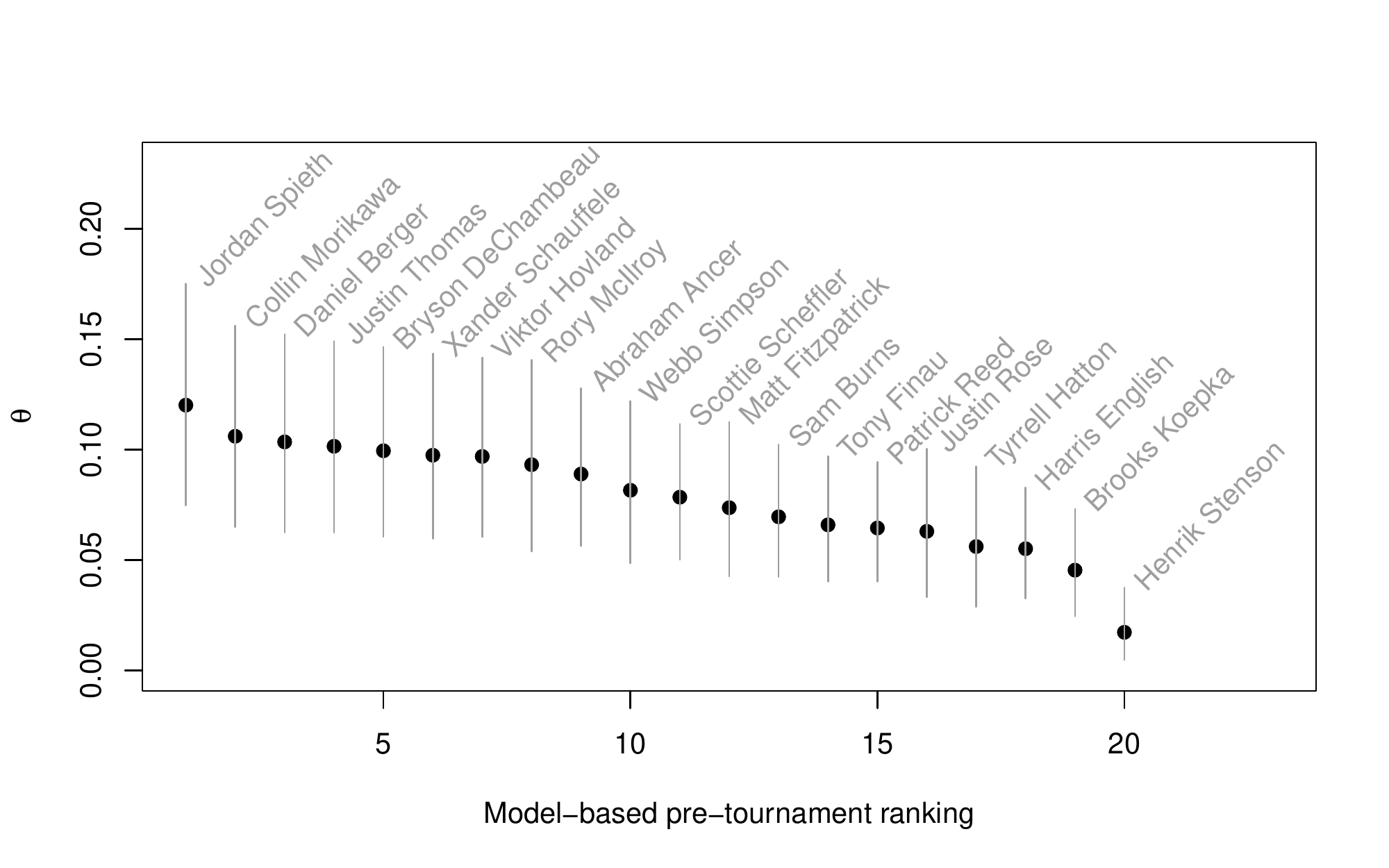}
\caption{Marginal posterior means (dots) and 95\% equal-tailed credible intervals (lines) for the
  parameters of the GPL model for the 20 players taking part in the
  2021 Hero World Challenge.}
\label{fig:uspga21-hero-post}
\end{figure}

It is of interest to use the GPL model to determine (prior) predictive
probabilities for each player winning the tournament. Such
probabilities can be compared directly to bookmakers odds to see if
the model is giving sensible predictions.   

Recall from Section~\ref{sec:int-prop} that $p_i=\frac{\theta_i
  \prod_{j\neq i}(1-\theta_j)}{1-\prod_{j}(1-\theta_j)}$ for
$i,j\in\mathcal{K}_p$, where $\mathcal{K}_p$ denotes the players in
this tournament, gives the probability under the GPL model of player
$i$ being outright winner (no ties). Conditioning on there being an
outright winner suggests $p_i/\sum_{j}p_j$ as a possible estimate of
the probability that player $i$ wins the tournament.  However, these
quantities are approximations due to the fact that in the actual
tournament ties for first place are broken by a play-off (extra holes
are played until an outright winner is found). The generative nature
of the GPL model allows a simulation-based estimate of the probability
of a win to be used which directly mimics the process used in the
actual tournament. Specifically, the outcome of the tournament is
simulated $N=10000$ times with values of the parameters sampled from
the posterior distribution; if there is a tie for first in the $j$th
simulated tournament then it is broken by simulating extra rank 
orderings for only the tied players, until there are no ties remaining.

Table~\ref{tab:uspga21-hero-pred} shows the final leaderboard from the
2021 Hero World Challenge together with simulation-based prior
predictive win probabilities calculated as described above. The GPL
model predicted that the tournament winner Viktor Hovland had a 7\%
chance of winning. This compares to the 9\% chance given to the
model-based pre-tournament favourite, Jordan Spieth. 
\begin{table}[tbh]
\centering
\begin{tabular}{rlrr}
\hline
Pos & Player & \multicolumn{2}{c}{Probability to win} \\
    &        &  GPL    & Market \\
\hline
1 &	Viktor Hovland	& 0.0676 & 0.0702\\
2 &	Scottie Scheffler & 0.0475 & 0.0496\\
T3 &	Sam Burns	& 0.0398 & 0.0496\\
T3 &	Patrick Reed	& 0.0350  & 0.0324\\
T5 &     Justin Thomas	& 0.0682 & 0.0766\\
T5 &	Collin Morikawa	& 0.0732 & 0.0842\\
T7 &	Tony Finau	& 0.0391 & 0.0496\\
T7 &	Daniel Berger	& 0.0691 & 0.0443\\
T9 &	Brooks Koepka	& 0.0224 & 0.0401\\
T9 &	Tyrrell Hatton	& 0.0314 & 0.0290\\
T9 &	Justin Rose	& 0.0312 & 0.0366\\
T12 &   Matt Fitzpatrick	& 0.0455 & 0.0366\\
T12 &	Xander Schauffele	& 0.0666 & 0.0648\\
T14 &	Bryson DeChambeau	& 0.0695 & 0.0648\\
T14 &	Abraham Ancer	&0.0553 & 0.0401\\
T14 &	Harris English	& 0.0313 & 0.0248\\
17  &	Webb Simpson	& 0.0510  & 0.0443\\
18  &	Rory McIlroy	& 0.0597 & 0.0936\\
19  &	Henrik Stenson	& 0.0077 & 0.0126\\
20  &	Jordan Spieth	& 0.0889 & 0.0562
\end{tabular}
\caption{Final leaderboard for the 2021 Hero World Challenge together
  with probabilities of a win based on the GPL model and pre-tournament
  market odds}
\label{tab:uspga21-hero-pred}
\end{table}
Bookmaker/market-based pre-tournament odds were also provided by \cite{cbs}; these have been
converted to probabilities and then normalised to sum to one for
direct comparison with the GPL model-based win probabilities as a means of external validation for the
model. The respective probabilities are listed in
Table~\ref{tab:uspga21-hero-pred} and show that the model-based
predictions are roughly in line with those from the bookmaker/market,
suggesting that the GPL model may form the basis of a model that
is suitable for prediction in this context. 

Simulations from the posterior predictive
distribution are also ideal for assessing the goodness of fit of the model \citep{BDA3}. Many possible features of the data could be
scrutinised and the choices made should be problem
dependent. For example, here it might be pertinent to assess whether
the patterns of ties in the simulated data are consistent with those
in the observed data. Figure~\ref{fig:uspga21-hero-sets-ties} displays
the predictive distribution of the number of ``buckets'' (or ordered
sets of ties)  in the $N=10000$ simulated realisations of the 2021 Hero World
Challenge. 
\begin{figure}[tbh]
\centering
\includegraphics[width=0.6\linewidth]{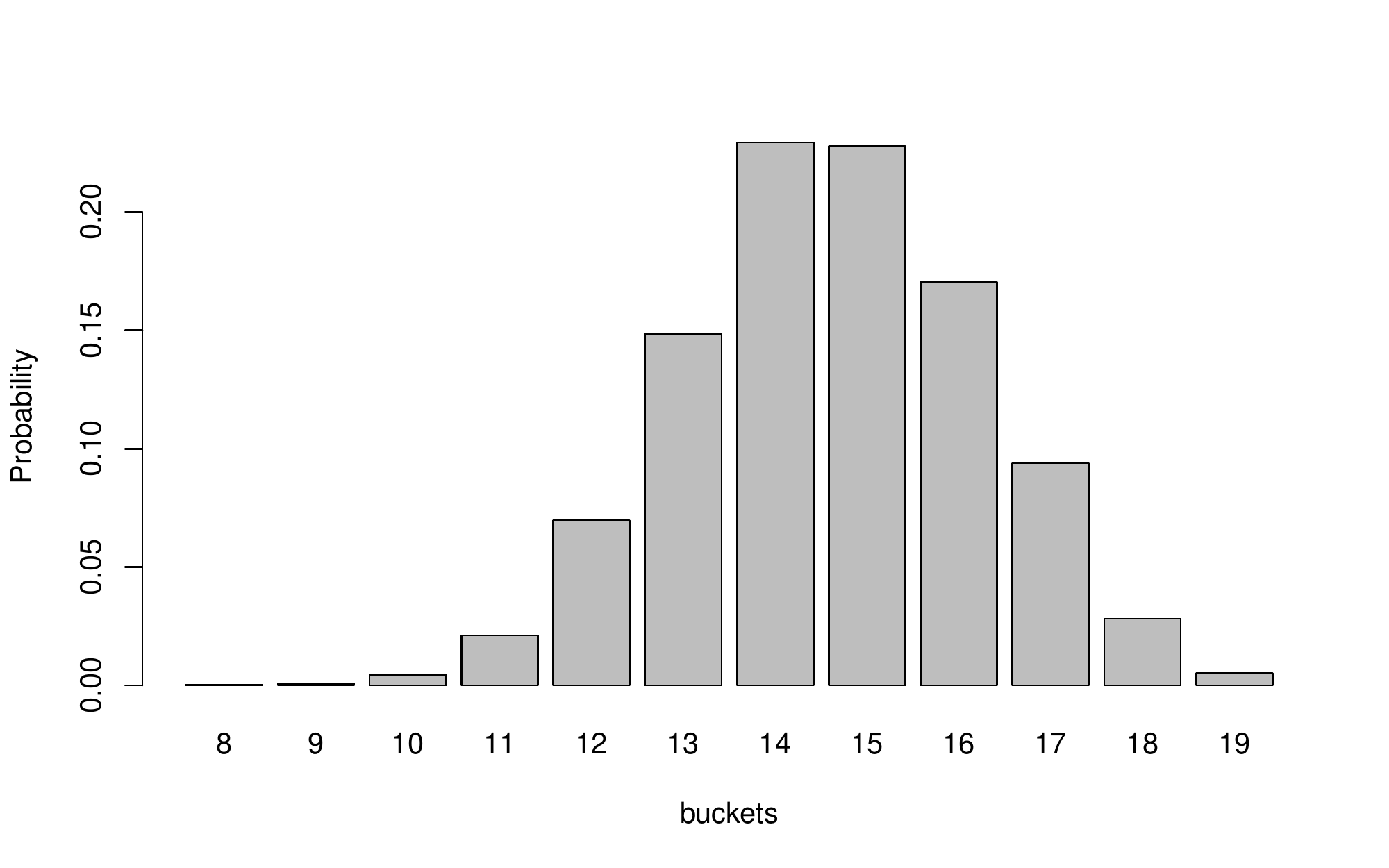}
\caption{Predictive distribution of the number of buckets of ties
  based on the $N=10000$ simulated realisations of the 2021  Hero World
Challenge.}
\label{fig:uspga21-hero-sets-ties}
\end{figure}
The actual observed value of 12 buckets has a predictive
probability of 0.067 and  lies
towards the lower tail of the predictive distribution, but is
nevertheless well supported and not unusual suggesting that for this
facet of the data at least, the GPL model is acceptable. Clearly there
are many other aspects of the data that could and should be subjected
to validation but the example presented here illustrates the general
point. 

\section{Discussion}
\label{sec:discuss}

The GPL model, as described in this paper, has been shown to be a
simple, intuitively appealing and natural model for rank ordered data
with ties, inheriting several desirable features from its continuous
counterpart the Plackett-Luce model. The GPL model was motivated as a
discrete analogue of the Plackett-Luce model and the Plackett-Luce
model was shown to be a limiting special case. The tractable nature of
the geometric latent variables in the GPL model allowed a closed form
expression for the likelihood to be derived. Simple data
augmentation-based algorithms for Bayesian inference were also
derived. The ease with which these algorithms could be implemented and
their efficiency was demonstrated on several real datasets covering a
range of data types. The real data examples showcased several features
of the GPL model, including its suitability for simulation-based
predictive inference. 

The GPL model is capable of fitting datasets of a size and complexity
that are currently beyond the reach of its close competitor the
Davidson-Luce model.  For example, the GPL model easily handled full
Bayesian inference on a model with over 600 parameters in the example
presented in Section~\ref{sec:golf}. Examining how well the GPL model
will scale up to massive datasets is an interesting avenue for future
work.

One of the main attractive features of the Plackett-Luce model is its
tractability, and this has allowed for many extensions of the basic
model. The GPL model can also be extended in several ways, and can be
used as a building block in more complex models. Natural extensions
are to mixtures of GPL models for heterogeneous datasets, inclusion of
covariates, and stage-specific variations (like that proposed by
\cite{Benter94} for the Plackett-Luce model). These extensions are the
subject of ongoing work. Whilst the main methodological focus of this
paper has been on the latent variable representation of the GPL model
and its use in data augmentation algorithms, the closed form
likelihood function opens the door to several of the aforementioned
and other possible extensions. In particular, the GPL model can be
conveniently coded in probabilistic programming languages such as
\textsf{Stan} and this should help to facilitate rapid uptake,
extensions and advancements in terms of both applications and
methodology.

\bibliographystyle{agsm}
\bibliography{geom}

\newpage

\begin{center}
\textbf{\LARGE Appendices}
\end{center}

\appendix

\section{Useful properties of geometric random variables}
\label{app:geom-prop}

Suppose that $X\sim
\text{Geom}(\theta_x)$, that is, the random variable $X$ has a geometric distribution with probability mass function
\[
\Pr(X=x)=(1-\theta_x)^{x-1}\theta_x, \quad 0<\theta_x\leq 1, \quad x=1,2,\ldots.
\]
It follows that $X$ has survival function
\[
S_X(x)=(1-\theta_x)^x, \quad x=1,2,\ldots.
\]

The following elementary properties are needed for the derivations presented in
the paper and are easily verified. 

If $Y\sim \text{Geom}(\theta_y)$ is independent of $X$ then 
\begin{align*}
\Pr(X = Y) &= \frac{\theta_x \theta_y}{1-(1-\theta_x)(1-\theta_y)}\\
\Pr(X < Y) &= \frac{\theta_x(1-\theta_y)}{1-(1-\theta_x)(1-\theta_y)}.
\end{align*}

If $X_i\overset{\text{indep.}}{\sim} \text{Geom}(\theta_i)$ for $i=1,\ldots,n$,
then 
$X_0=\min \{X_1,\ldots,X_n\}\sim \text{Geom}(1-\prod_{i=1}^n(1-\theta_i))$.

\section{Bayesian analysis of the puddings data using the Davidson model}
\label{sec:puddings-davidson}

For comparison with the results of the GPL models in Section~\ref{sec:puddings-GPL} and to put the
results in context, a Bayesian analysis of the puddings data using the
Davidson model (Section~\ref{sec:Davidson}) was also performed. Independent gamma prior
distributions were adopted for the worth parameters, that is,
$\lambda_k\overset{\text{indep.}}{\sim}\text{Gamma}(a_\lambda,b_\lambda)$
for $k=1,\ldots,K$, with $a_\lambda=b_\lambda=1$ for the results presented herein. The prior distribution avoids the necessity to
implement a constraint on the $\lambda_k$ to ensure identifiability.
The prior distribution on the ties parameter $\delta$ is also gamma,
that is $\delta\sim\text{Gamma}(a_\delta,b_\delta)$, with in this
example, $a_\delta=b_\delta=1$. The Davidson model (see Section~\ref{sec:Davidson}) was
implemented in \texttt{rtan}~\citep{rstan} an \textsf{R} interface to \textsf{Stan}~\citep{stan22v2.31} and Figure~\ref{fig:pud-dav-post}
summarises the marginal posterior distributions of the parameters
based on 10000 values sampled from the posterior distribution.
\begin{figure}[h]
\centering
\includegraphics[width=0.6\linewidth]{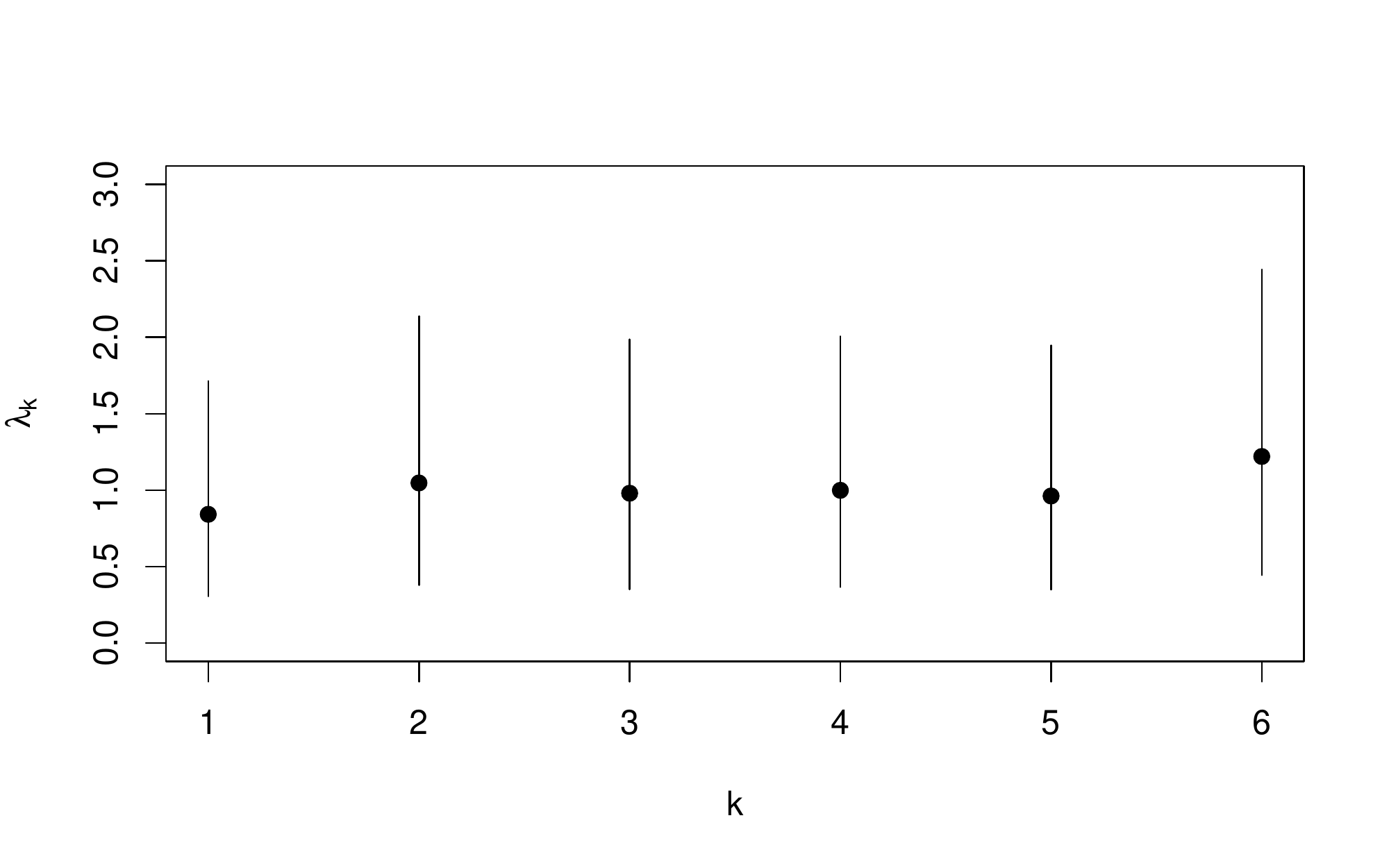}
\caption{Marginal posterior means (dots) and 95\% equal-tailed
  credible intervals (lines) for the worth 
  parameters of the Davidson model applied to the puddings data.}
\label{fig:pud-dav-post}
\end{figure}
As in the analyses under the GPL model, there is considerable overlap
in the posterior credible intervals of the $\lambda$ parameters.
Ordering the posterior means gives an overall rank ordering of the chocolate
puddings of
$\hat{\vec{y}}=(6,2,4,3,5,1)'$, the same as that under the reverse GPL
model. Perhaps this is not surprising given the similarity between the
Davidson model and the reverse GPL model; see Section~\ref{sec:Davidson}.  The posterior
mean and equal-tailed 95\% credible interval (in parentheses) for the
ties parameter $\delta$ is $0.753~(0.639, 0.879)$. These results are
comparable (up to scaling) with the approximate MLEs reported in
\cite{Davidson70} and \cite{turner2020modelling}; see Figure~1 in
\cite{turner2020modelling}. 

Finally, the LOOIC is 1631.4 (8.1) for the
Bayesian Davidson model demonstrating that the fit of the model is
comparable to that of the two GPL models. Further chi-squared goodness
of fit measures (not included) also show that the fit of the models
are very similar. For comparison, the AIC value reported in
\cite{turner2020modelling} for their non-Bayesian analysis is 1631.4.

\section{Golf data and additional information}
\label{sec:golf-info}

Consider the data at
\url{www.owgr.com/events/149th-open-championship-8187}, relating to
the 149th Open Championship, to illustrate the format of the data
analyses in this paper. The only data recorded were the values in the
finishing position column and the golfer's name column. For example,
Will Zalatoris withdrew after the first round and so is not included
in the final rank ordering.

Tournaments in the dataset are predominantly stroke play although
match play tournaments (for example, the Dell Technologies Match Play
Championship) are included. Given that the GPL model only depends on
the finishing order and not the actual score it can easily handle both
types of results.  Stroke play tournaments are typically contested
over four rounds of 18 holes, in which the golfer tries to take as few
shots as possible to get the ball in the hole, hence the lowest score
over the 72 holes wins. A match play tournament involves players
scoring points per hole, one for a win and half for a tie, where,
again, the fewest shots taken to get the ball in the hole wins. Scores
are discrete and so ties are common. Ties in first place are also
possible; for example, the Tour Championship on 5th September 2021
ended in a tie for first between Kevin Na and Jon Rahm. However, ties
for first place are typically broken by playing extra holes until an
outright winner emerges.

A common feature of golf tournaments is the ``cut'' in which the field
is reduced to the top players after 36 of the 72 holes have been
completed. For example, on the PGA Tour, it is typical for the field
to be cut to the top 65 players including ties after 36 holes. The
players who ``make the cut'' continue to play the remaining 36 holes
whereas those that ``miss the cut'' take no further part.  The results
are therefore in the form of top-$m_i$ rank orderings, where $m_i$
depends on the tournament $i$. Rather than ignore the information from
the players that missed the cut it is important to incorporate this
information correctly; they are not tied in (for example) 66th, but
they are ranked lower than 65th. 

\end{document}